\documentclass{colt2018} 

\usepackage{amsfonts,color}
\usepackage{fancybox}
\usepackage{tikz}
\usepackage{hyperref}
\usepackage{cleveref} 
\usepackage{thm-restate}
\crefformat{equation}{(#2#1#3)}

\usepackage{caption}
\usepackage{bbm}
\usepackage{tablefootnote}

\SetAlCapSkip{.5em}

\newtheorem{fact}[theorem]{Fact}

\newenvironment{fminipage}%
{\begin{Sbox}\begin{minipage}}%
		{\end{minipage}\end{Sbox}\fbox{\TheSbox}}

\def\prob#1#2{\mbox{Pr}_{#1}\left[ #2 \right]}

\def\expec#1#2{{\mathbb{E}}_{#1}\left[ #2 \right]}

\def\defeq{\stackrel{\mathrm{def}}{=}}

\def\sgn#1{\mathrm{sgn} (#1)}

\def\abs#1{\left|#1  \right|}

\def\norm#1{\left\| #1 \right\|}

\def\calA{\mathcal{A}}

\newcommand\bb{\boldsymbol{\mathit{b}}}

\newcommand\pp{\boldsymbol{\mathit{p}}}

\newcommand\vv{\boldsymbol{\mathit{v}}}
\newcommand\ww{\boldsymbol{\mathit{w}}}
\newcommand\yy{\boldsymbol{\mathit{y}}}

\newcommand\xx{\boldsymbol{\mathit{x}}}

\renewcommand\AA{\boldsymbol{\mathit{A}}}
\newcommand\BB{\boldsymbol{\mathit{B}}}

\newcommand\II{\boldsymbol{\mathit{I}}}

\newcommand\MM{\boldsymbol{\mathit{M}}}

\renewcommand\SS{\boldsymbol{\mathit{S}}}

\newcommand\UU{\boldsymbol{\mathit{U}}}
\newcommand\WW{\boldsymbol{\mathit{W}}}

\newcommand\YY{\boldsymbol{\mathit{Y}}}
\newcommand\ZZ{\boldsymbol{\mathit{Z}}}

\newcommand\Otil{\widetilde{O}}

\newcommand\lw{{\overline{\boldsymbol{w}}}}
\newcommand\LW{{\overline{\boldsymbol{\mathit{W}}}}}
\newcommand\barA{{\overline{\boldsymbol{\mathit{A}}}}}

\newcommand{\comment}[1]{}

\newcommand{\norme}[1]{\norm{#1}_2}
\newcommand{\R}{\mathbb{R}}
\newcommand{\ones}{\mathbbm{1}}

\title[$\ell_1$ Regression using Lewis Weights and SGD]{$\ell_1$ Regression using Lewis Weights Preconditioning and Stochastic Gradient Descent}
\usepackage{times}

\coltauthor{\Name{David Durfee} \Email{ddurfee@gatech.edu}\\
	\Name{Kevin A. Lai} \Email{kevinlai@gatech.edu}\\
	\Name{Saurabh Sawlani} \Email{sawlani@gatech.edu}\\
	\addr Georgia Institute of Technology}

\begin{document}

\maketitle

\begin{abstract}
We present preconditioned stochastic gradient descent (SGD) algorithms for the $\ell_1$ minimization problem $\min_{\xx}\|\AA \xx - \bb\|_1$ in the overdetermined case, where there are far more constraints than variables. Specifically, we have $\AA \in \R^{n \times d}$ for $n \gg d$. Commonly known as the Least Absolute Deviations problem, $\ell_1$ regression can be used to solve many important combinatorial problems, such as minimum cut and shortest path. SGD-based algorithms are appealing for their simplicity and practical efficiency.
Our primary insight is that careful preprocessing can yield preconditioned matrices $\tilde{\AA}$ with strong properties (besides good condition number and low-dimension) that allow for faster convergence of gradient descent. In particular, we precondition using Lewis weights to obtain an isotropic matrix with fewer rows and strong upper bounds on all row norms. We leverage these conditions to find a good initialization, which we use along with recent smoothing reductions and accelerated stochastic gradient descent algorithms to achieve $\epsilon$ relative error in $\Otil(nnz(\AA) + d^{2.5} \epsilon^{-2})$ time with high probability, where $nnz(\AA)$ is the number of non-zeros in $\AA$. This improves over the previous best result using gradient descent for $\ell_1$ regression. We also match the best known running times for interior point methods in several settings.

Finally, we also show that if our original matrix $\AA$ is approximately isotropic and the row norms are approximately equal, we can give an algorithm that avoids using fast matrix multiplication and obtains a running time of $\Otil(nnz(\AA) + s d^{1.5}\epsilon^{-2} + d^2\epsilon^{-2})$, where $s$ is the maximum number of non-zeros in a row of $\AA$. In this setting, we beat the best interior point methods for certain parameter regimes.

\end{abstract}

\begin{keywords}
$\ell_1$ regression, stochastic gradient descent, Lewis weights
\end{keywords}

\section{Introduction}\label{sec:intro}


Stochastic gradient descent (SGD) is one of the most widely-used practical algorithms for optimization problems due to its simplicity and practical efficiency
\citep{NesterovV08, NemirovskiJLS09}. We consider SGD methods to solve the unconstrained overdetermined $\ell_1$ regression problem,
commonly known as the Least Absolute Deviations problem,
which is defined as follows:

\begin{align}
\min_{\xx \in \mathbb{R}^d}  \norm{\AA \xx - \bb}_1,\label{eq:l1}
\end{align}
where $\AA \in \mathbb{R}^{n \times d}$, $\bb \in \mathbb{R}^{n}$ and $n \gg d$. Compared to Least Squares ($\ell_2$) regression, the $\ell_1$ regression problem is more robust and is thus useful when outliers are present in the data. Moreover, many important combinatorial problems, such as minimum cut or shortest path,
can be formulated as $\ell_1$ regression problems \citep{CMMP13}, and high accuracy $\ell_1$ regression can be used to solve general linear programs.\footnote{For instance, one can determine if $\{\xx|\AA \xx - \bb, \xx \ge 0\}$ is feasible by writing an objective of the form $\alpha(\norm{\AA \xx - \bb}_1 + \norm{\xx}_1 + \norm{\xx - \beta \ones}_1)$ where $\alpha$ and $\beta$ are sufficiently large polynomials in the input size.} Since \cref{eq:l1} can be formulated as a linear program \citep{PortnoyK97, ChenDS01}, generic methods for solving linear programs, such as the interior-point method (IPM), can be used to solve it \citep{PortnoyK97, Portnoy97, MengM13mapreduce, LeeS15}.

SGD algorithms are popular in practice for $\ell_1$ and other regression problems because they are simple, scalable, and efficient. State-of-the-art algorithms for solving \cref{eq:l1} utilize sketching techniques from randomized linear algebra to achieve $\text{poly}(d,\epsilon^{-1})$ running times, whereas a naive extension of Nesterov acceleration \citep{Nesterov83} to certain classes of non-smooth functions \citep{Nesterov05smooth, Nesterov05gap, Nesterov07, AllenZhuH16} takes $\text{poly}(n,\epsilon^{-1})$ time. This difference is significant because $n \gg d$ in the overdetermined setting. Ideally, the only dependence on $n$ in the running time will be implicitly in an additive $nnz(\AA)$ term.  

Sketching techniques from randomized numerical linear algebra look to find a low-distortion embedding of $\AA$ into a smaller subspace, after which popular techniques for $\ell_1$ regression can be applied on the reduced matrix. Efforts to build these sampled matrices or ``coresets"
have been made using random sampling \citep{Clarkson05},
fast matrix multiplication \citep{SohlerW11},
and ellipsoidal rounding \citep{DasguptaDHKM09, ClarksonDMMMW13}.
All of these methods produce coresets of size 
$\text{poly}(d,\epsilon^{-1}) \times d$ in time $O(n \cdot \text{poly}(d))$.
\cite{MengM13embedding} and \cite{WoodruffZ13} improve these
techniques to produce similar coresets in $O \left( nnz(\AA)+\text{poly}(d) \right)$ time.


In addition to using sketching as a preprocessing step, one can also apply \emph{preconditioners}. Preconditioners transform the input matrix into one with additional desirable properties, such as good condition number, that speed up subsequent computation. For our setting, we will use the term ``preconditioning'' to refer to dimensionality reduction followed by additional processing of the matrix. Of particular interest for our setting is the preconditioning technique of \cite{cohenpeng}, which utilizes a Lewis change of density \citep{Lewis78} to sample rows of $\AA$ with probability proportional to their Lewis weights such that the sampled matrix $\tilde{\AA}$ approximately preserves $\ell_1$ distances, which is to say that $\norm{\AA \xx}_1 \approx ||\tilde{\AA}\xx||_1$ for any vector $\xx$. Lewis weights are in essence the ``correct'' sampling weights for dimensionality reduction in $\ell_1$ regression, and they are used by the previous best SGD-based methods for solving \cref{eq:l1} \citep{YangCRM16}. As it turns out, Lewis weights also lead to nice $\ell_2$ conditions for the sampled matrix. One of the key insights of this paper is to leverage these additional guarantees to obtain significantly faster running times for SGD after Lewis weight preconditioning.

\paragraph{Techniques} \label{subsec:results}

Our techniques for solving the $\ell_1$ regression problem follow the popular paradigm of preconditioning and then using gradient descent methods on the resulting problem. Typically, the preconditioner is a black-box combination of a sketching method with a matrix rotation, which yields a well-conditioned low-dimensional matrix. The crucial idea in this paper is that the sketch can give us some strong properties in addition to the low-dimensional embedding. By more tightly integrating the components of the preconditioner, we obtain faster running times for $\ell_1$ regression.

In particular, preconditioning the given matrix-vector pair $[\AA\;\bb]$ using Lewis weights \citep{cohenpeng}
achieves a low-dimensional $[\tilde{\AA}\;\tilde{\bb}]$ such that
$\norm{\AA \xx - \bb}_1 \approx \|\tilde{\AA} \xx - \tilde{\bb}\|_1$
for every $\xx \in \mathbb{R}^{d}$. 
It is then possible to apply the low-dimensional embedding properties of Lewis weights in a black-box manner to $\ell_1$ regression algorithms, using the fact that this embedding reduces the row-dimension from $n$ to $O(d \epsilon^{-2}\log n)$ and then plugging these new parameters into the runtime guarantees.
Our critical observation will be that sampling by Lewis weights also has the important property that all leverage scores of the new matrix are approximately equal.
Since rotations of a matrix do not change its leverage scores, we are free to rotate $\tilde{\AA}$ to place it into isotropic position, in which case the leverage score condition implies that the row norms are tightly bounded.

The isotropic and row norm conditions yield two essential phenomena. First, a careful choice of initial vector can be shown to be close to optimal. Second, we get strong bounds on the gradient of any row. Using these properties, it is almost immediately implied that standard SGD only requires $O(d^2\epsilon^{-2})$ iterations to arrive at a solution $\hat{x}$  with relative error\footnote{Relative error $\epsilon$ means that $f(\hat{\xx}) - f(\xx^*) \le \epsilon f(\xx^*)$, where $f(\xx) = \AA \xx - \bb$ and $\xx^* = \arg\min_{\xx} f(\xx)$.} $\epsilon$, leading to a total running time of $\Otil(nnz(\AA) + d^3 \epsilon^{-2})$.
These properties can be further applied to smoothing reductions by \cite{AllenZhuH16} and accelerated SGD
algorithms by \cite{AllenZhu17} to improve the runtime to
$\Otil(nnz(\AA) + nd^{\omega - 1} + \sqrt{n}d^{1.5} \epsilon^{-1})$. As previously mentioned, sampling by Lewis weights guarantees that $n \le O(d\epsilon^{-2}\log{n})$, so we also obtain a running time of $\Otil(nnz(\AA) + d^{2.5} \epsilon^{-2})$.
\Cref{fig:mainAlgo} gives the basic framework of our $\ell_1$ solver.

\begin{algorithm2e}[h]
	\caption{General structure of our algorithm \label{fig:mainAlgo}}
	\KwIn{Matrix $\AA \in \R^{n\times d}$, and vector $\bb \in \R^{n}$, along with error parameter $\epsilon > 0$.}
	\begin{enumerate*}
		\item Precondition $[\AA\;\bb]$ by Lewis weight sampling as in \cite{cohenpeng}, along with a matrix rotation.
		\item Initialize $\xx_0$ as the exact or approximate  $\ell_2$ minimizer of the preconditioned problem.
		\item Run a stochastic gradient descent algorithm on the preconditioned problem with starting point $\xx_0$.
	\end{enumerate*}
\end{algorithm2e}

\paragraph{Our Results}
Our first main theorem uses smoothing reductions from \cite{AllenZhuH16} with the accelerated SGD algorithm of \cite{AllenZhu17}:
\begin{restatable}[]{theorem}{applyKatyusha}
  \label{thm:applyKatyusha}
  Given $\AA \in \R^{n\times d}$, $\bb \in \R^n$, assume $\min_{\xx} \norme{\AA \xx - \bb}$ is either 0 or bounded below by $n^{-c}$ for some constant $c>0$. Then for any $\epsilon > 0$, there is a routine that outputs $\tilde{\xx}$ such that with high probability,\footnote{Throughout this paper, we let ``with high probability" mean that $\xi$ is the failure probability and our runtime has dependence on $\log(1/\xi)$, which we ignore for ease of notation.} \[
  \norm{\AA\tilde{\xx} - \bb}_1 \leq (1 + \epsilon)\min_{\xx} \norm{\AA\xx - \bb}_1
  \] 
  with a runtime of $O\left(nnz(\AA)\log^2{n} + d^{2.5}\epsilon^{-2}\log^{1.5}{n}\right)$ whenever $n \geq d\epsilon^{-2}\log{n}$,
  and a runtime of
  $O\left( \sqrt{n}d^2 \epsilon^{-1}\log{n}  + nd^{\omega - 1}\log{n}\right)$ when $n \leq d\epsilon^{-2}\log{n}$.
\end{restatable}
\noindent Achieving the bounds when $n\le d\epsilon^{-2}\log n$ requires some additional technical work. \Cref{thm:applyKatyusha} is proved in \Cref{sec:SGDforl1}, where we also show our $\Otil(nnz(\AA) + d^3\epsilon^{-2})$ running time for standard SGD.


Our second main theorem is motivated by the fact that the theoretical running time bounds of fast matrix multiplication
are difficult to achieve in practice, and most implementations of algorithms actually use naive matrix multiplication. Thus, it would be ideal for an algorithm's running time to be independent of fast matrix multiplication. It turns out that our only dependence on fast matrix multiplication is during the preconditioning stage. 
Accordingly, if we are given a matrix which is already approximately isotropic with all row norms approximately equal, then we can eliminate the usage
of fast matrix multiplication and still prove the same time bound. Moreover, this method preserves the row-sparsity of $\AA$. The primary difficulty of this approach is in computing an appropriate initialization when $\AA$ is only approximately isotropic. To resolve this issue, we use efficient $\ell_2$ regression solvers that do not rely on fast matrix multiplication.

\begin{restatable}[]{theorem}{noFastMatrixMult}
  \label{thm:noFastMatrixMult}
  Let $\AA \in \R^{n\times d}$ and $\bb \in \R^{n}$ be such that the matrix-vector pair $[\AA\;\bb]$ satisfies $[\AA\;\bb]^T [\AA\;\bb]\approx_{O(1)}\II$ \footnote{As defined in the notation section, we say that $\AA \approx_{\kappa} \BB$ if and only if
  $	\dfrac{1}{\kappa} \BB \preceq \AA \preceq \kappa \BB$.} and for each row $i$ of $[\AA\;\bb]$, $\norme{[\AA\;\bb]_{i,:}}^2 \approx_{O(1)} d/n$. Assume $\norme{\bb}\le n^c$ and $\min_{\xx} \norme{\AA \xx - \bb}$ is either 0 or bounded below by $n^{-c}$ for some constant $c > 0$. Then for any $\epsilon > 0$, there is a routine that computes $\tilde{\xx}$ such that with high probability,
  \[
  \norm{\AA\tilde{\xx} - \bb}_1 \leq (1 + \epsilon)\min_{\xx}\norm{\AA\xx - \bb}_1
  \]
  with a runtime of $O\left(nnz(\AA)\log^2{n} + s \cdot d^{1.5}\epsilon^{-2} \log^{1.5}{n} + d^2 \epsilon^{-2} \log^2{n} \right)$, where $s$ is the maximum number of entries in any row of $\AA$.  
\end{restatable}

\noindent The added assumption on $\norme{b}$ in \Cref{thm:noFastMatrixMult} comes from the bounds of the $\ell_2$ solvers. We prove \Cref{thm:noFastMatrixMult} in \Cref{sec:sparsitypreserve}. 


\paragraph{Comparison of our results with previous work}

Our algorithms are significantly faster than the previous best SGD-based results for $\ell_1$ regression, which took $O(nnz(\AA)\log n+ d^{4.5}\epsilon^{-2}\sqrt{\log d})$ time \citep{YangCRM16}. Furthermore, our standard SGD bounds are especially likely to be achievable in practice. \Cref{table:regression} compares the running time of our algorithm to the fastest gradient descent methods \citep{Clarkson05, Nesterov09, YangCRM16}, interior point methods \citep{MengM13mapreduce, LeeS15}, and multiplicative weights update methods \citep{CMMP13}. Since we can apply sampling by Lewis weights prior to any algorithm\footnote{Sampling $\AA$ by Lewis weights creates a dense $(d\epsilon^{-2}\log n) \times d$ matrix $\tilde{\AA}$}, we can replace $n$ with $O(d\epsilon^{-2} \log n)$ and $nnz(\AA)$ with $O(d^2 \epsilon^{-2} \log n)$ at the cost of an additive $\Otil(nnz(\AA) + d^\omega)$ overhead for any running time bound in \Cref{table:regression}, where $d^\omega$ is time to multiply two $d\times d$ matrices.\footnote{The current best value for $\omega$ is approximately $2.373$ \citep{Williams12, DavieS13, LeGall14}.}

Note that we match the theoretical performance of the current best IPM \citep{LeeS15} in several regimes. In low to medium-precision ranges, for example $\epsilon \ge 10^{-3}$, both the best IPM and our algorithm have a running time of $\Otil(nnz(\AA) + d^{2.5})$. If all of the algorithms are implemented with naive matrix multiplication, \cite{LeeS15} takes $\Otil(nnz(\AA) + d^3)$ time, while we prove an identical running time for standard non-accelerated SGD with Lewis weights preconditioning. For general $\epsilon$, if $nnz(\AA) \ge d^2 \epsilon^{-2} \log n$, then \cite{LeeS15} will use Lewis weights sampling and both their algorithm and our algorithm will achieve a running time of $\Otil(nnz(\AA) + d^{2.5}\epsilon^{-2})$. This is significant because our SGD-based algorithms are likely to be far more practical\footnote{Lewis weights preconditioning is also fast in practice}. Finally, we also note that in the setting where $\AA$ is approximately isotropic with approximately equal row norms and $\AA$ is row-sparse, with at most $s$ non-zeros per row, our algorithm in \Cref{thm:noFastMatrixMult} has the best dimensional-dependence out of any existing algorithm for $s < d$. In particular, we beat \cite{LeeS15} whenever $s < d\epsilon^{2}$ or whenever $s<d$ and $nnz(\AA) \ge d^2 \epsilon^{-2} \log n$.

\renewcommand{\arraystretch}{1.3} 

\begin{table}[h]
	\centering
	\begin{tabular}{| l | l |}
		\hline
		\multicolumn{1}{|c|}{\bfseries Solver} & \multicolumn{1}{c|}{\bfseries Running time \tablefootnote{$\Otil$ hides terms polylogarithmic in $d$ and $n$.} } \\
		\hline      
		Subgradient descent \cite{Clarkson05} & $\Otil \left(nd^5 \epsilon^{-2} \log (1/\epsilon) \right)$ \\
		\hline      
		Smoothing and gradient descent \cite{Nesterov09} & $\Otil \left(n^{1.5}d \epsilon^{-1} \right)$  \\
		\hline      
		Randomized IPCPM\tablefootnote{Interior Point Cutting Plane Methods} \cite{MengM13mapreduce} & $\Otil \left(nd^{\omega-1} +  \left( nnz(\AA) +  \text{poly}(d) \right)d \log \left( 1 /\epsilon \right) \right)$ \\
		\hline
		Multiplicative weights \cite{CMMP13} & $\Otil(n^{1/3}(nnz(\AA) + d^2)\epsilon^{-8/3})$\\
		\hline
		IPM and fast inverse maintenance \cite{LeeS15} & $\Otil ( ( nnz(\AA) + d^2 ) \sqrt{d} \log ( 1/\epsilon ) )$\\
		\hline      
		Weighted SGD \cite{YangCRM16} & $\Otil(nnz(\AA) + d^{4.5} \epsilon^{-2})$ \\
		\hline      
		Lewis weights and SGD (this work) & $\Otil \left( nnz(\AA) + d^3 \epsilon^{-2} \right)$ \\
		\hline      
		Lewis weights and accelerated SGD (this work)\tablefootnote{This running time only assumes that $\omega \leq 2.5$} & $\Otil\left( nnz(\AA) + d^{2.5}\epsilon^{-2} \right)$
		\\
		\hline  
	\end{tabular}
	\caption{Running time of several iterative $\ell_1$ regression algorithms.
		All running times are to find a solution with
		$\epsilon$ relative error, with constant failure probability.
		Note that the first three algorithms in the table could also be sped up
		using the preconditioning method we use, i.e.,
		Lewis weights row sampling \citep{cohenpeng}.
		Doing this would need a preconditioning time of $\Otil(nnz(\AA) + d^{\omega})$,
		and enable us to use the fact that $n \le O(d \epsilon^{-2} \log{n} )$ after preconditioning.
		However, this still does not make them faster than later algorithms.
		\label{table:regression}}
\end{table}

Another related work by \cite{BubeckCLL17:arxiv} gives algorithms for $\ell_p$ regression for $p \in (1,\infty)$ that run in time $\Otil(nnz(\AA) \cdot n^{|1/2 - 1/p|} \log(1 / \epsilon) )$ and $\Otil(nnz(\AA) n^{|1/2 - 1/p| - 1/2} d^{1/2} + n^{|1 /2  - 1/p|} d^{2} + d^{\omega}) \log(1 / \epsilon) )$. They also use preconditioning and accelerated stochastic gradient descent techniques as a subroutine. However, they don't give bounds for $\ell_1$ regression. Also, for $p$ close to 1, these bounds are worse than ours. In contrast to \cite{BubeckCLL17:arxiv}, our algorithms are specific to the $\ell_1$ setting. We use the fact that the gradients are bounded for $\ell_1$ regression, which doesn't hold for general $\ell_p$ regression. Moreover, our initialization using an $\ell_2$ minimizer doesn't give good bounds for general $\ell_p$ regression. In this sense, our algorithms really leverage the special structure of the $\ell_1$ problem.

%

\subsection{Organization}
The paper is organized as follows.
\Cref{sec:prelims} contains definitions and basic lemmas
which we will use throughout the paper.
\Cref{sec:SGDforl1} contains our main contribution, i.e.,
once we are given a suitably preconditioned matrix,
it shows how we arrive at an approximate $\ell_1$
minimizer within the claimed time bounds,
for both non-accelerated and accelerated versions
of stochastic gradient descent.
\Cref{sec:sparsitypreserve} shows that
if we restrict our input to slightly weaker
preconditions, then we can eliminate the need for
fast matrix multiplication to achieve the same time bounds.
\Cref{sec:katyushaproof} contains the primary proof details
for our main result in \Cref{sec:SGDforl1}.
In \Cref{sec:lewis}, we show that row sampling using
Lewis weights \cite{cohenpeng},
along with matrix rotation, suffices to give us
a matrix satisfying our precondition requirements.
\Cref{sec:sec4proofs} contains proof
details from \Cref{sec:sparsitypreserve}.
Finally, \Cref{sec:minor_details} will give secondary and straightforward technical details for our main results that we include for completeness.
\section{Preliminaries}\label{sec:prelims}

In this section, we describe some of the notation and important definitions we use in the paper. We represent matrices and vectors using bold variables. We let $\AA_{i,:}$ denote the $i^{th}$ row of a matrix $\AA$, and we use $nnz(\AA)$ to denote the number of non-zero elements in $\AA$.
$\AA^{\dagger}$ refers to the Moore-Penrose pseudoinverse of $\AA$. When $\AA$ has linearly-independent columns,
$\AA^{\dagger} = (\AA^T \AA)^{-1} \AA^T$.
Also, we assume that the input $\AA$ has full rank.

\begin{definition}[$\ell_p$-norm]
	The $\ell_p$ norm of a vector $\vv \in \R^n$
	is defined as
	\begin{align*}
	\textstyle{\norm{\vv}_p \defeq \left(\sum_{i=1}^n \vv_i^p \right)^{1/p}}.
	\end{align*} 
	Accordingly, the $\ell_p$ norm of a matrix $\AA \in \R^{n \times d}$
	is defined as
	\begin{align*}
	\norm{\AA}_p \defeq \sup_{\xx \in \R^d, \xx \neq 0} \dfrac{\norm{\AA \xx}_p}{\norm{\xx}_p}.
	\end{align*} 
\end{definition}

\begin{definition}[Matrix approximation]
	We say that $\AA \approx_{\kappa} \BB$ if and only if
	\begin{align*}
	\dfrac{1}{\kappa} \BB \preceq \AA \preceq \kappa \BB.
	\end{align*}
	Here, $\preceq$ refers to the L\"{o}wner partial ordering of matrices,
	where we say that $\AA \preceq \BB$ if $\BB - \AA$ is positive semi-definite.
\end{definition}

Note that we also use $\approx$ similarly in the case of scalars, as is commonplace.

\begin{definition}[IRB]\label{def:good} A matrix $\AA \in \R^{n\times d}$ with $n \ge d$ is said to be \textit{isotropic row-bounded} (IRB) if the following hold:
	\begin{enumerate}
		\item{$\AA^T \AA = \II$},
		\item{For all rows of $\AA$, $\norme{\AA_{i,:}}^2 \leq O(d/n)$}.
	\end{enumerate}
\end{definition}

\begin{definition}\label{def:levScore}
	Given a matrix $\AA$, we define the \textit{statistical leverage score} of row $\AA_{i,:}$ to be 
	\[\tau_i(\AA) \defeq \AA_{i,:}\left(\AA^T \AA\right)^{-1}\AA_{i,:}^T = \norme{\left(\AA^T\AA\right)^{-1/2}\AA_{i,:}^T}^2. \]
\end{definition}

\begin{definition}\label{def:lewis}
	For a matrix $\AA$, the $\ell_1$ \textit{Lewis weights} $\lw$ are the unique weights such that for each row $i$ we have 
	\[\lw_i = \tau_i\left(\LW^{-1/2}\AA\right) 
	\]
	or equivalently
	\[\lw^2 = \AA_{i,:}\left(\AA^T\LW^{-1}\AA\right)^{-1}\AA_{i,:}^T\]
	where $\LW$ is the diagonal matrix formed by putting the elements of $\lw$ on the diagonal.
\end{definition}

\begin{definition}\label{def:smooth}
	A function $f$ is $L$-\textit{smooth} if for any $x,y \in \mathbb{R}^d$, $||\nabla f(x) - \nabla f(y)||_2 \leq L||x-y||_2.$
\end{definition}

\begin{definition}\label{def:strong_convex}
	A function $f$ is $\sigma$-\textit{strongly convex} if for any $x,y \in \mathbb{R}^d$,
	\begin{equation*}
	\textstyle{f(y) \geq f(x) + \langle \nabla f(x),y-x\rangle + \frac{\sigma}{2}||x-y||_2^2}.
	\end{equation*}
\end{definition}

\begin{definition}\label{def:lipsCont}
	A function $f$ is $G$-\textit{Lipschitz continuous} if for any $x,y \in \mathbb{R}^d$, 	
	\[ ||f(x) - f(y)||_2 \leq G||x-y||_2.\] 	
\end{definition}

\section{Stochastic Gradient Descent for $\ell_1$ Regression}\label{sec:SGDforl1}

In this section, we describe how we achieve the bounds in \Cref{thm:applyKatyusha}. We first introduce the preconditioning technique by \cite{cohenpeng}, which, along with rotating the matrix,
will reduce our problem to $\ell_1$ minimization where the input matrix $\AA$ is isotropic and the norms of all its rows have strong upper bounds, i.e. it is IRB by Definition~\ref{def:good}. We relegate the details and proof of this preconditioning procedure to \Cref{sec:lewis}.

In \Cref{subsec:sgd}, we prove that known stochastic gradient descent algorithms
will run provably faster if we assume that $\AA$ is IRB. In particular, if $\AA$ is IRB, we can find an initialization $\xx_0$ that is close to the optimum $\xx^*$,
which in addition to bounding the gradient of our objective,
will then allow us to plug these bounds into standard stochastic gradient descent algorithms
and achieve a runtime of $\Otil(nnz(\AA) + d^3\epsilon^{-2})$. Finally, in \Cref{subsec:katyusha} we take known smoothing techniques by \cite{AllenZhuH16}
along with the Katyusha accelerated stochastic gradient descent by \cite{AllenZhu17} to achieve a runtime of $\Otil(nnz(\AA) + d^{2.5}\epsilon^{-2})$.

\subsection{Preconditioning with Lewis weights}
\label{sec:LewisPreconditioning}
The primary tool in our preconditioning routine will be a sampling scheme by Lewis weights, introduced in \cite{cohenpeng}, that was shown to approximately preserve the $\ell_1$ norm.
Specifically, we will use the combination of two primary theorems from \cite{cohenpeng}
that approximately compute the Lewis weights of a matrix quickly
and then sample accordingly while still approximately
preserving $\ell_1$ norm distances with high probability.

\begin{theorem}[Theorem 2.3 and 6.1 from \cite{cohenpeng}]
	\label{thm:lewisWeights}
	Given a matrix $\AA \in \R^{n \times d}$
	with $\ell_1$ Lewis weights $\lw$ and
	an error parameter $\epsilon > 0$,
	then for any function $h(n,\epsilon) \geq O(\epsilon^{-2}\log{n})$,
	we can find sampling values
	\[\pp_i \approx_{O(1)} \lw_i h(n,\epsilon)\]
	for each $i \in \{1,2, \ldots, n\}$,
	and generate a matrix $\SS$ with $N = \sum_i\pp_i$ rows,
	each chosen independently as the $i^{th}$ standard basis vector of dimension $n$,
	times $\frac{1}{\pp_i}$ with probability proportional to $\pp_i$,
	such that with high probability we have 
	\[ \textstyle \|\tilde{\AA}\xx\|_1 \approx_{1+\epsilon}\norm{\AA\xx}_1\]
	for all $\xx \in \R^d$, where $\tilde{\AA} = \SS\AA$. 
	Computing these sampling values requires
	$O(nnz(\AA)\log{n} + d^{\omega})$ time.
\end{theorem} 

In \Cref{sec:lewis}, we will show that this sampling scheme also ensures that
each row of $\tilde{\AA}$ has approximately the same leverage score.
This proof will involve applying known facts about leverage scores and their connections with Lewis weights, along with standard matrix concentration bounds.
Furthermore, we will obtain additional nice properties by rotating $\AA$
and showing that a solution to our reduced problem gives an approximate solution to the original problem,
culminating in the following lemma:

\begin{restatable}{lemma}{lewisAndRotate}
	\label{lem:lewisAndRotate}
	There is a routine that takes
	a matrix $\AA \in \R^{n \times d}$,	a vector $\bb \in \R^{n}$ and $\epsilon > 0$,
	then produces a matrix $[\tilde{\AA},\tilde{\bb}] \in \R^{N \times (d+1)}$
	with $N = O(d\epsilon^{-2}\log{n})$
	and an invertible matrix $\UU \in \R^{d \times d}$
	such that matrix $\tilde{\AA}\UU$ is IRB and
	if $\tilde{\xx}_{\UU}^*$ minimizes $\|\tilde{\AA}\UU\xx - \tilde{b}\|_1$,
	then for any $\tilde{\xx}$ such that 
	\[ \|\tilde{\AA}\UU\tilde{\xx} - \tilde{\bb}\|_1\leq (1 + \delta) \|\tilde{\AA}\UU\tilde{\xx}_{\UU}^* - \tilde{\bb}\|_1, \]
	holds for some $\delta > 0$, we must have 
	\[
	\|\AA(\UU\tilde{\xx}) - \bb\|_1 \leq (1 + \epsilon)^2(1 + \delta)\|\AA \xx^* - \bb \|_1 \]
	with high probability.
	
	Furthermore, the full running time is
	$O(nnz(\AA)\log{n} + d^{\omega-1}\min\{d\epsilon^{-2}\log{n},n\} + \Upsilon)$ where $\Upsilon = \min\{d\epsilon^{-2}\log{n}, (d\epsilon^{-2}\log{n})^{1/2 + o(1)} + n\log^2{n}\}$.
\end{restatable}

As a result, we will assume that all of our matrices $\AA$ are already in the same form as $\tilde{\AA}\UU$, i.e. we assume $\AA$ is IRB, since relative error guarantees for the preconditioned system apply to the original system.

\subsection{Isotropic and Row-Bounded $\AA$ for Stochastic Gradient Descent}\label{subsec:sgd}

To demonstrate the usefulness of the properties of our preconditioned $\AA$,
we consider standard stochastic gradient descent and the bounds on its running time. We let $\xx^* = \arg\min_{\xx} \norm{\AA\xx - \bb}_1$. We will use the following theorem to prove runtime bounds for standard SGD:

\begin{theorem}[\cite{RS86ssgd}]
	\label{thm:standardSGD}
	Given a function $f$ and $\xx_0$ such that
	$\norme{\xx_0 - \xx^*} \leq R$
	and $L$ is an upper bound on the $\ell_2$ norm of the stochastic subgradients of $f$,
	then projected subgradient descent ensures that after $t$ steps:
	\[
	f(\xx_t) - f(\xx^*) \leq O\left(\frac{RL}{\sqrt{t}}\right).
	\]
	where $\xx^* = \arg \min_{\xx} f(\xx)$.
\end{theorem}

To use this theorem, we must prove bounds on the initialization distance $\norme{\xx_0 -\xx^*}$ and the norm of the stochastic subgradients we use, i.e. $\nabla (n \cdot |\AA_{i,:}\xx - \bb_i|)$ for each $i$. We show that both of these bounds come from our assumptions on $\AA$.

\begin{lemma}\label{lem:easyInit}
	If $\AA$ is IRB, then by setting $\xx_0 = \AA^T\bb$ we have
	\[
	\norme{\xx_0 - \xx^*}^2 \leq O\left(d/n\right)\norm{\AA\xx^*-\bb}_1^2.
	\]
\end{lemma}

\begin{proof}
	\begin{align*}
	\textstyle{\norm{\xx^{\left( 0 \right)} - \xx^{\left( * \right)}}_{2}^2} & = \textstyle{\norm{\AA^T \bb - \xx^{\left( * \right)}}_{2}^2} \\
	& \textstyle{= \norm{\AA^T \left(\bb - \AA\xx^{(*)}\right)}_2^2} & &\text{by assumption $\AA^T\AA = \II$}\\
	& \textstyle{= \norme{\sum_i \AA_{i,:} \left(\bb - \AA\xx^{(*)}\right)}^2} \\
	& \textstyle{\leq \norm{\AA\xx^*-\bb}_1^2 \cdot \max_i ||\AA_{i,:}||_2^2} & & \text{by convexity of $\norm{\cdot}_2$, also shown in Lemma~\ref{lem:subordinatenorm}}\\
	& = O\left({\dfrac{d}{n}}\right)\norm{\AA\xx^*-\bb}_1^2.
	\end{align*}
\end{proof}

\comment{ The convexity argument (Jensen) does rely on the $1$-norm. So if we relax it back to $p$-norm for some $p > 1$, we lose a factor of $n^{p - 1}$, giving
\begin{align}
	\norme{\xx_0 - \xx^*}^2
	& = \norme{\AA^T\left( \xx_0 - \xx^* \right) }^2\\
	& \leq \norm{\AA\xx^*-\bb}_1^2 \cdot \frac{d}{n}\\
	& \leq O\left(d n^{2p - 3} \right)\norm{\AA\xx^*-\bb}_{p}^2.
\end{align}
Specifically if $\AA \xx^{*} - \bb$ is the all-$1$s vector,
this is actually tight.
On the other hand, for $p = 2$,
we have that $\AA \xx^{*} - \bb$ is orthogonal to $\AA^{T}$,
and should be $0$.
So this kind of algebra seems only to be tight in the $p = 1$ case.
}
\begin{lemma}\label{lem:easyLipscitz}
		If $\AA$ is IRB, then $\norme{\nabla (n \cdot |\AA_{i,:}\xx - \bb_i|)}^2 \leq O(nd)$ for all $i$.
\end{lemma}


\begin{proof} We see that $\nabla (n \cdot |\AA_{i,:}\xx - \bb_i|) = n \cdot \AA_{i,:}^T\sgn{\AA_{i,:}\xx - \bb_i} $, and $\norme{\AA_{i,:}}^2 \leq O(d/n)$ for all $i$ by our assumption that $\AA$ is IRB. This then implies our desired inequality.
\end{proof}

These bounds, particularly the initialization distance, are stronger than the bounds for general $\AA$, and together will give our first result that improves upon the runtime given by \cite{YangCRM16} by using our preconditioning.

\begin{restatable}[]{theorem}{ourSGD}
	\label{thm:ourSGD}
	Given $\AA \in \R^{n\times d}$, we can find $\tilde{\xx}\in \R^d$ using preconditioning and stochastic gradient descent such that 
	\[
	\norm{\AA\tilde{\xx} - \bb}_1 \leq (1 + \epsilon)\norm{\AA \xx^* - \bb}_1
	\]	
	in time $O(nnz(\AA)\log^2{n} + d^3\epsilon^{-2}\log{n})$.
\end{restatable}

\begin{proof}
	By preconditioning with Lemma~\ref{lem:lewisAndRotate} and error $O(\epsilon)$ we obtain an $N \times d$ matrix $\tilde{\AA}\UU$ in time $O(nnz(\AA)\log{n} + d^{\omega}\epsilon^{-2})$.
	
	By \Cref{thm:standardSGD}, we then need to run $O(d^2 \epsilon^{-2})$ iterations of standard stochastic gradient descent to achieve absolute error of $O(\epsilon\cdot f(\xx^*))$ which is equivalent to relative error of $O(\epsilon)$. The required runtime is then $O(d^3\epsilon^{-2})$. Technically, the input to stochastic gradient descent will require the value $R$, i.e. the upper bound on initialization distance, which requires access to a constant factor approximation of $f(\xx^*)$. We will show in \Cref{subsec:binarySearch} that we can assume that we have such an approximation at the cost of a factor of $\log{n}$ in the running time.
	
	Combining the preconditioning and stochastic gradient descent will produce $\tilde{\xx}$ with $O(\epsilon)$ relative error to the optimal objective function value in time $O(nnz(\AA)\log{n} + d^3\epsilon^{-2})$. Adding the factor $\log{n}$ overhead from estimating $f(\xx^*)$ gives the desired runtime.
\end{proof}

\subsection{Smoothing Reductions and Katyusha Accelerated SGD}
\label{subsec:katyusha}
In order to improve the running time guarantees, we consider whether our strong initialization distance bound will allow us to apply black-box accelerated stochastic gradient descent methods. These methods generally require smoothness and strong convexity of the objective function, neither of which are necessarily true for our objective function. Previous results \citep{Nesterov05smooth, Nesterov07, DuchiBW12, OuyangG12, AllenZhuH16} have addressed this general issue and given reductions from certain classes of objective functions to similar functions with smoothness and strong convexity, while still maintaining specific error and runtime guarantees. Accordingly, we will first show how our initialization distance fits into the reduction of \cite{AllenZhuH16}, then apply Katyusha's accelerated gradient descent algorithm by \cite{AllenZhu17} to their framework. 
We state the theorem below, and relegate the details of its proof to \Cref{sec:katyushaproof}.

\applyKatyusha*

The last bound in \Cref{thm:applyKatyusha} may seem odd, as Lewis weights sampling produces a matrix with row-dimension $d\epsilon^{-2}\log{n}$, which will in fact hurt our running time if $n \ll d\epsilon^{-2}\log{n}$. Moreover, we cannot simply avoid running the Lewis weights sampling because our algorithm critically relies on the resulting leverage score properties.
Instead, if $n \ll d\epsilon^{-2}\log{n}$, we can simulate the sampling procedure in $O(n)$ time and keep a count for each of the $n$ unique rows. Since the simulated sample matrix will look like the original but with duplicated rows, we will be able to carry out the rest of our linear algebraic manipulations in time dependent on $n$ rather than $d\epsilon^{-2}\log{n}$. 
We will address this in \Cref{subsec:simulatedSplitting}.


\section{Row-Sparsity Bounds for $\ell_1$ Regression}\label{sec:sparsitypreserve}
In this section, we explain how to avoid using matrix multiplication, which we use in Lemma~\ref{lem:lewisAndRotate}, to get an IRB matrix, and in Lemma~\ref{lem:easyInit}, to get a good initialization. To avoid both procedures, we assume that our given matrix $\AA \in \R^{n\times d}$ and vector $\bb \in \R^{n }$ are such that $[\AA\;\bb]^T [\AA\;\bb] \approx_{O(1)} \II$ and for each row $i$ of $[\AA\;\bb]$ we have $\norme{[\AA\;\bb]_{i,:}}^2 \approx_{O(1)} d/n$. Notice that these conditions imply $\AA^T\AA \approx_{O(1)} \II$ and $\norme{\AA_{i,:}}^2 \leq O(d/n)$, which are nearly the properties of the preconditioned matrix $\tilde{\AA}\UU$ generated in Lemma~\ref{lem:lewisAndRotate}. However, we still face two new complications: (1) the row count of matrix $\AA$ has not been reduced from $n$ to $O(d\epsilon^{-2}\log{n})$, and (2) $\AA$ is only approximately isotropic.

We will account for the first complication by showing that, under these conditions, the Lewis weights are approximately equal, which implies that uniform row sampling is nearly equivalent to that in \Cref{thm:lewisWeights}. 
We then describe how to find a good initialization point using conjugate gradient methods when $\tilde{\AA}$ is only approximately isotropic. 
Finally, we use the reduction in \cite{AllenZhuH16} and the Katyusha stochastic gradient descent algorithm from \cite{AllenZhu17} to achieve a total running time of $\Otil(nnz(\AA) + sd^{1.5}\epsilon^{-2} + d^2\epsilon^{-2})$, where $s$ is the maximum number of non-zeros in a row of $\AA$. Note that as a byproduct of this algorithm, we achieve row-sparsity-preserving $\ell_1$ minimization.

\subsubsection*{Uniform Sampling of $\AA$}
We deal with the first complication of avoiding Lemma~\ref{lem:lewisAndRotate} by sampling $\AA$ uniformly. In particular, if we uniformly sample $N = O(d\epsilon^{-2} \log d)$ rows of $[\AA\;\bb]$, this yields a smaller matrix $[\tilde{\AA}\;\tilde{\bb}]$ such that $\tilde{\AA}^T\tilde{\AA} \approx_{O(1)} \II$ and $\|\tilde{\AA}_{i,:}\|_2^2 \leq  O(d/N)$ for each row $i$, which is to say $\tilde{\AA}$ is almost IRB. This then culminates in the following lemma whose proof we relegate to \Cref{subsec:unifSample}.

\begin{restatable}[]{lemma}{uniformSampling}
	\label{lem:uniformSampling}
	Suppose we are given a matrix $\AA \in \R^{n \times d}$ such that $\AA^T\AA \approx_{O(1)} \II$ and $\norme{\AA_{i,:}}^2 \approx_{O(1)} d/n$.
	If we uniformly sample $N = O(d\epsilon^{-2}\log{n})$ rows independently and rescale each row by $n/N$ to obtain matrix $\tilde{\AA}$, then with high probability the following properties hold:
	\begin{enumerate}
		\item $\norm{\AA\xx}_1 \approx_{1 + \epsilon} \|\tilde{\AA}\xx\|_1$ for all $\xx \in \R^d$.
		\item $\tilde{\AA}^T\tilde{\AA} \approx_{O(1)} \left(\dfrac{n}{N}\right)\II$.
		\item $\|\tilde{\AA}_{i,:}\|_2^2 \approx_{O(1)} dn/N^2$ for all rows $i \in \{1,2,\ldots, N\}$.
	\end{enumerate}
\end{restatable}


\subsubsection*{Initialization using Approximate $\ell_2$ Minimizer}\label{subsec:l2Minimizer}
It now remains to show that we can still find a good initialization $\xx_0$ for gradient descent even with our relaxed assumptions on $\AA$.
Previously, when we had $\AA^T \AA = \II$, we used $\xx_0 = \AA^T \bb = \arg \min_{\xx} \norme{\AA \xx - \bb}$.
It turns out that for $\AA^T \AA \approx_{O(1)} \II$, the $\ell_2$ minimizer $\xx_0 = \arg \min_{\xx} \norme{\AA \xx - \bb}$ is still a good initialization point.
But finding an exact $\ell_2$ minimizer would take a prohibitive amount of time or would require matrix multiplication.
However, an approximate $\ell_2$ minimizer suffices, and we can find such a point quickly using the conjugate gradient method.

For this section, we define $\xx^* \defeq \arg\min_{\xx} \norm{\AA\xx - \bb}_1$ and $\xx_0 \defeq \arg \min_{\xx} \norme{\AA\xx - \bb}$. Our main result is the following:

\begin{lemma}[Approximate $\ell_2$ minimizer is close to $x^*$] \label{lem:initMain} Let $\AA \in \R^{n \times d}$ be such that $\AA^T \AA \approx_{O(1)} \II$ and for each row $i$ of $\AA$, $\norme{\AA_{i,:}}^2 \leq O(d/n)$. Assume that $\norme{\bb} \le n^c$ and $\norme{\AA\xx_0 - \bb} \ge n^{-c}$ for some constant $c > 0$.
	\begin{align*}
	\norme{\tilde{\xx}_0 - \xx^*} \le O(\sqrt{d/n}) \norm{\AA \xx^* - \bb}_1.
	\end{align*}
	Moreover, $\tilde{\xx}_0$ can be computed in  $O( (t_{\AA^T \AA} + d) \log (n/ \epsilon))$ time, where $t_{\AA^T \AA}$ denotes the time to multiply a vector by $\AA^T \AA$. 
\end{lemma}

We prove Lemma~\ref{lem:initMain} using the following two lemmas whose proofs are deferred until \Cref{subsec:initialization_proofs}. Lemma~\ref{lem:exactInit} shows that the $\ell_2$ minimizer is close to the $\ell_1$ minimizer even when $\AA$ is only approximately isotropic, and the proof is similar to that of Lemma~\ref{lem:easyInit}.
Lemma~\ref{lem:cg} shows that we can find a good estimate for the $\ell_2$ minimizer.

\begin{lemma}[Exact $\ell_2$ minimizer is close to $\xx^*$]\label{lem:exactInit}
	Let $\AA \in \R^{n \times d}$ be such that $\AA^T \AA \approx_{O(1)} \II$ and for each row $i$ of $\AA$, $\norme{\AA_{i,:}}^2 \leq O(d/n)$. Then
	\begin{align*}
	\norme{\xx_0 - \xx^*} \le O(\sqrt{d/n}) \norm{\AA \xx^* - \bb}_1.
	\end{align*}
\end{lemma}

\begin{lemma}[Conjugate gradient finds an approximate $\ell_2$ minimizer]\label{lem:cg}
Let $\AA \in \R^{n \times d}$ be such that $\AA^T \AA \approx_{O(1)} \II$ and for each row $i$ of $\AA$, $\norme{\AA_{i,:}}^2 \leq O(d/n)$. Assume that $\norme{\bb} \le n^c$ and $\norme{\AA\xx_0 - \bb} \ge n^{-c}$.
Then for any $\epsilon > 0$, the conjugate gradient method can find an $\tilde{\xx}_0$ such that
\begin{align*}
\norme{\tilde{\xx}_0 - \xx_0} \le \epsilon \norme{\AA\xx_0 - \bb}.
\end{align*}
Moreover, $\tilde{\xx}_0$ can be found in time $O( (t_{\AA^T \AA} + d) \log (n/ \epsilon))$, where $t_{\AA^T \AA}$ is the time to multiply a vector by $\AA^T \AA$. 
\end{lemma}

\begin{proof}[Proof of Lemma~\ref{lem:initMain}]
	We use conjugate gradient with $\epsilon = \sqrt{d/n}$ to find an $\tilde{\xx}_0$ in $O((t_{\AA^T \AA} + d) \log (n/ \epsilon))$ time by Lemma~\ref{lem:cg} to achieve our desired initialization time bounds.
	Note that by definition of $\xx_0$ and by a standard norm inequality, we have:
	\begin{align*}
	\norme{\AA \xx_0 - \bb} \le \norme{\AA \xx^* - \bb} \le \norm{\AA \xx^* - \bb}_1
	\end{align*}
	
	Then by the triangle inequality and Lemma~\ref{lem:exactInit} we have:
	\begin{align*}
	\norme{\tilde{\xx}_0 - \xx^*} &\le \norme{\tilde{\xx}_0 - \xx_0} + \norme{\xx_0 - \xx^*} \\
	&\le \sqrt{d/n}\norm{\AA \xx^* - \bb}_1 + O(\sqrt{d/n})\norm{\AA \xx^* - \bb}_1
	\end{align*}
	which gives our desired initialization accuracy.
\end{proof}

\subsubsection*{Fast Row-sparsity-preserving $\ell_1$ Minimization}

Finally, we combine the matrix achieved by uniform sampling in Lemma~\ref{lem:uniformSampling} and the initialization from Lemma~\ref{lem:initMain} to achieve fast row-sparsity-preserving $\ell_1$ minimization. This then gives the following main theorem that we prove in \Cref{subsec:achievingSparsity}.

\noFastMatrixMult*


\acks{We are grateful to Richard Peng for many insightful discussions. We would also like to thank Michael Cohen for his helpful comments
	and Junxing Wang for his feedback. We thank the anonymous reviews of an earlier version of this
	manuscript for their very valuable comments, especially with
	regards to the performances of the many existing
	algorithms in the $\ell_1$ case.}

\bibliography{gd}

\begin{thebibliography}{37}
\providecommand{\natexlab}[1]{#1}
\providecommand{\url}[1]{\texttt{#1}}
\expandafter\ifx\csname urlstyle\endcsname\relax
  \providecommand{\doi}[1]{doi: #1}\else
  \providecommand{\doi}{doi: \begingroup \urlstyle{rm}\Url}\fi

\bibitem[{Allen Zhu}(2017)]{AllenZhu17}
Zeyuan {Allen Zhu}.
\newblock Katyusha: the first direct acceleration of stochastic gradient
  methods.
\newblock In \emph{Proceedings of the 49th Annual {ACM} {SIGACT} Symposium on
  Theory of Computing, {STOC} 2017, Montreal, QC, Canada, June 19-23, 2017},
  pages 1200--1205, 2017.
\newblock \doi{10.1145/3055399.3055448}.
\newblock URL \url{http://doi.acm.org/10.1145/3055399.3055448}.

\bibitem[Allen-Zhu and Hazan(2016)]{AllenZhuH16}
Zeyuan Allen-Zhu and Elad Hazan.
\newblock Optimal black-box reductions between optimization objectives.
\newblock In \emph{Proceedings of the 30th International Conference on Neural
  Information Processing Systems}, NIPS'16, pages 1614--1622, USA, 2016. Curran
  Associates Inc.
\newblock ISBN 978-1-5108-3881-9.
\newblock URL \url{http://dl.acm.org/citation.cfm?id=3157096.3157277}.

\bibitem[Bubeck et~al.(2017)Bubeck, Cohen, Lee, and Li]{BubeckCLL17:arxiv}
S{\'e}bastien Bubeck, Michael~B Cohen, Yin~Tat Lee, and Yuanzhi Li.
\newblock An homotopy method for $\ell_p$ regression provably beyond
  self-concordance and in input-sparsity time.
\newblock \emph{arXiv preprint arXiv:1711.01328}, 2017.
\newblock Available at: https://arxiv.org/abs/1711.01328.

\bibitem[Chen et~al.(2001)Chen, Donoho, and Saunders]{ChenDS01}
Scott~Shaobing Chen, David~L. Donoho, and Michael~A. Saunders.
\newblock Atomic decomposition by basis pursuit.
\newblock \emph{SIAM Rev.}, 43\penalty0 (1):\penalty0 129--159, January 2001.
\newblock ISSN 0036-1445.
\newblock \doi{10.1137/S003614450037906X}.
\newblock URL \url{http://dx.doi.org/10.1137/S003614450037906X}.

\bibitem[Chin et~al.(2013)Chin, Madry, Miller, and Peng]{CMMP13}
Hui~Han Chin, Aleksander Madry, Gary~L. Miller, and Richard Peng.
\newblock Runtime guarantees for regression problems.
\newblock In \emph{Proceedings of the 4th Conference on Innovations in
  Theoretical Computer Science}, ITCS '13, pages 269--282, New York, NY, USA,
  2013. ACM.
\newblock ISBN 978-1-4503-1859-4.
\newblock \doi{10.1145/2422436.2422469}.
\newblock URL \url{http://doi.acm.org/10.1145/2422436.2422469}.

\bibitem[Clarkson(2005)]{Clarkson05}
Kenneth~L. Clarkson.
\newblock Subgradient and sampling algorithms for $\ell_1$ regression.
\newblock In \emph{Proceedings of the Sixteenth Annual ACM-SIAM Symposium on
  Discrete Algorithms}, SODA '05, pages 257--266, Philadelphia, PA, USA, 2005.
  Society for Industrial and Applied Mathematics.
\newblock ISBN 0-89871-585-7.
\newblock URL \url{http://dl.acm.org/citation.cfm?id=1070432.1070469}.

\bibitem[Clarkson et~al.(2013)Clarkson, Drineas, Magdon-Ismail, Mahoney, Meng,
  and Woodruff]{ClarksonDMMMW13}
Kenneth~L. Clarkson, Petros Drineas, Malik Magdon-Ismail, Michael~W. Mahoney,
  Xiangrui Meng, and David~P. Woodruff.
\newblock The fast cauchy transform and faster robust linear regression.
\newblock In \emph{Proceedings of the Twenty-fourth Annual ACM-SIAM Symposium
  on Discrete Algorithms}, SODA '13, pages 466--477, Philadelphia, PA, USA,
  2013. Society for Industrial and Applied Mathematics.
\newblock ISBN 978-1-611972-51-1.
\newblock URL \url{http://dl.acm.org/citation.cfm?id=2627817.2627851}.

\bibitem[Cohen and Peng(2015)]{cohenpeng}
Michael~B. Cohen and Richard Peng.
\newblock $\ell_p$ row sampling by lewis weights.
\newblock In \emph{Proceedings of the Forty-seventh Annual ACM Symposium on
  Theory of Computing}, STOC '15, pages 183--192, New York, NY, USA, 2015. ACM.
\newblock ISBN 978-1-4503-3536-2.
\newblock \doi{10.1145/2746539.2746567}.
\newblock URL \url{http://doi.acm.org/10.1145/2746539.2746567}.

\bibitem[Cohen et~al.(2015)Cohen, Lee, Musco, Musco, Peng, and
  Sidford]{CohenLMMPS15}
Michael~B. Cohen, Yin~Tat Lee, Cameron Musco, Christopher Musco, Richard Peng,
  and Aaron Sidford.
\newblock Uniform sampling for matrix approximation.
\newblock In \emph{Proceedings of the 2015 Conference on Innovations in
  Theoretical Computer Science}, ITCS '15, pages 181--190, New York, NY, USA,
  2015. ACM.
\newblock ISBN 978-1-4503-3333-7.
\newblock \doi{10.1145/2688073.2688113}.
\newblock URL \url{http://doi.acm.org/10.1145/2688073.2688113}.

\bibitem[Dasgupta et~al.(2009)Dasgupta, Drineas, Harb, Kumar, and
  Mahoney]{DasguptaDHKM09}
Anirban Dasgupta, Petros Drineas, Boulos Harb, Ravi Kumar, and Michael~W.
  Mahoney.
\newblock Sampling algorithms and coresets for $\ell_p$ regression.
\newblock \emph{SIAM J. Comput.}, 38\penalty0 (5):\penalty0 2060--2078,
  February 2009.
\newblock ISSN 0097-5397.
\newblock \doi{10.1137/070696507}.
\newblock URL \url{http://dx.doi.org/10.1137/070696507}.

\bibitem[Davie and Stothers(2013)]{DavieS13}
A.~M. Davie and A.~J. Stothers.
\newblock Improved bound for complexity of matrix multiplication.
\newblock \emph{Proceedings of the Royal Society of Edinburgh: Section A
  Mathematics}, 143\penalty0 (2):\penalty0 351--369, 2013.
\newblock \doi{10.1017/S0308210511001648}.

\bibitem[Duchi et~al.(2012)Duchi, Bartlett, and Wainwright]{DuchiBW12}
John~C Duchi, Peter~L Bartlett, and Martin~J Wainwright.
\newblock Randomized smoothing for stochastic optimization.
\newblock \emph{SIAM Journal on Optimization}, 22\penalty0 (2):\penalty0
  674--701, 2012.

\bibitem[Farach-Colton and Tsai(2015)]{Farach-Colton2015}
Mart\'{\i}n Farach-Colton and Meng-Tsung Tsai.
\newblock Exact sublinear binomial sampling.
\newblock \emph{Algorithmica}, 73\penalty0 (4):\penalty0 637--651, December
  2015.
\newblock ISSN 0178-4617.
\newblock \doi{10.1007/s00453-015-0077-8}.
\newblock URL \url{http://dx.doi.org/10.1007/s00453-015-0077-8}.

\bibitem[Foster(1953)]{Foster53}
F.~G. Foster.
\newblock On the stochastic matrices associated with certain queuing processes.
\newblock \emph{The Annals of Mathematical Statistics}, 24\penalty0
  (3):\penalty0 355--360, 1953.
\newblock ISSN 00034851.
\newblock URL \url{http://www.jstor.org/stable/2236286}.

\bibitem[Harvey(2012)]{Harvey12}
Nick Harvey.
\newblock Matrix concentration and sparsification.
\newblock Workshop on ``Randomized Numerical Linear Algebra (RandNLA): Theory
  and Practice", 2012.
\newblock Available at:
  http://www.drineas.org/RandNLA/slides/Harvey\_RandNLA@FOCS\_2012.pdf.

\bibitem[Le~Gall(2014)]{LeGall14}
Fran\c{c}ois Le~Gall.
\newblock Powers of tensors and fast matrix multiplication.
\newblock In \emph{Proceedings of the 39th International Symposium on Symbolic
  and Algebraic Computation}, ISSAC '14, pages 296--303, New York, NY, USA,
  2014. ACM.
\newblock ISBN 978-1-4503-2501-1.
\newblock \doi{10.1145/2608628.2608664}.
\newblock URL \url{http://doi.acm.org/10.1145/2608628.2608664}.

\bibitem[Lee and Sidford(2015)]{LeeS15}
Yin~Tat Lee and Aaron Sidford.
\newblock Efficient inverse maintenance and faster algorithms for linear
  programming.
\newblock In \emph{Proceedings of the 2015 IEEE 56th Annual Symposium on
  Foundations of Computer Science (FOCS)}, FOCS '15, pages 230--249,
  Washington, DC, USA, 2015. IEEE Computer Society.
\newblock ISBN 978-1-4673-8191-8.
\newblock \doi{10.1109/FOCS.2015.23}.
\newblock URL \url{http://dx.doi.org/10.1109/FOCS.2015.23}.

\bibitem[Lewis(1978)]{Lewis78}
D.~Lewis.
\newblock Finite dimensional subspaces of $\ell_{p}$.
\newblock \emph{Studia Mathematica}, 63\penalty0 (2):\penalty0 207--212, 1978.
\newblock URL \url{http://eudml.org/doc/218208}.

\bibitem[Meng and Mahoney(2013{\natexlab{a}})]{MengM13embedding}
Xiangrui Meng and Michael~W. Mahoney.
\newblock Low-distortion subspace embeddings in input-sparsity time and
  applications to robust linear regression.
\newblock In \emph{Proceedings of the Forty-fifth Annual ACM Symposium on
  Theory of Computing}, STOC '13, pages 91--100, New York, NY, USA,
  2013{\natexlab{a}}. ACM.
\newblock ISBN 978-1-4503-2029-0.
\newblock \doi{10.1145/2488608.2488621}.
\newblock URL \url{http://doi.acm.org/10.1145/2488608.2488621}.

\bibitem[Meng and Mahoney(2013{\natexlab{b}})]{MengM13mapreduce}
Xiangrui Meng and Michael~W. Mahoney.
\newblock Robust regression on mapreduce.
\newblock In \emph{Proceedings of the 30th International Conference on
  International Conference on Machine Learning - Volume 28}, ICML'13, pages
  III--888--III--896. JMLR.org, 2013{\natexlab{b}}.
\newblock URL \url{http://dl.acm.org/citation.cfm?id=3042817.3043036}.

\bibitem[Nemirovski et~al.(2009)Nemirovski, Juditsky, Lan, and
  Shapiro]{NemirovskiJLS09}
A.~Nemirovski, A.~Juditsky, G.~Lan, and A.~Shapiro.
\newblock Robust stochastic approximation approach to stochastic programming.
\newblock \emph{SIAM J. on Optimization}, 19\penalty0 (4):\penalty0 1574--1609,
  January 2009.
\newblock ISSN 1052-6234.
\newblock \doi{10.1137/070704277}.
\newblock URL \url{http://dx.doi.org/10.1137/070704277}.

\bibitem[Nesterov and Vial(2008)]{NesterovV08}
Yu. Nesterov and J.~Ph. Vial.
\newblock Confidence level solutions for stochastic programming.
\newblock \emph{Automatica}, 44\penalty0 (6):\penalty0 1559--1568, June 2008.
\newblock ISSN 0005-1098.
\newblock \doi{10.1016/j.automatica.2008.01.017}.
\newblock URL \url{http://dx.doi.org/10.1016/j.automatica.2008.01.017}.

\bibitem[Nesterov(1983)]{Nesterov83}
Yurii Nesterov.
\newblock A method of solving a convex programming problem with convergence
  rate ${O}(1/k^2)$.
\newblock In \emph{Soviet Mathematics Doklady}, volume~27, pages 372--376,
  1983.

\bibitem[Nesterov(2005{\natexlab{a}})]{Nesterov05gap}
Yurii Nesterov.
\newblock Excessive gap technique in nonsmooth convex minimization.
\newblock \emph{SIAM Journal on Optimization}, 16\penalty0 (1):\penalty0
  235--249, May 2005{\natexlab{a}}.
\newblock ISSN 1052-6234.
\newblock \doi{10.1137/S1052623403422285}.
\newblock URL \url{http://dx.doi.org/10.1137/S1052623403422285}.

\bibitem[Nesterov(2005{\natexlab{b}})]{Nesterov05smooth}
Yurii Nesterov.
\newblock Smooth minimization of non-smooth functions.
\newblock \emph{Mathematical Programming}, 103\penalty0 (1):\penalty0 127--152,
  May 2005{\natexlab{b}}.
\newblock ISSN 0025-5610.
\newblock \doi{10.1007/s10107-004-0552-5}.
\newblock URL \url{http://dx.doi.org/10.1007/s10107-004-0552-5}.

\bibitem[Nesterov(2007)]{Nesterov07}
Yurii Nesterov.
\newblock Smoothing technique and its applications in semidefinite
  optimization.
\newblock \emph{Mathematical Programming}, 110\penalty0 (2):\penalty0 245--259,
  March 2007.
\newblock ISSN 0025-5610.
\newblock \doi{10.1007/s10107-006-0001-8}.
\newblock URL \url{http://dx.doi.org/10.1007/s10107-006-0001-8}.

\bibitem[Nesterov(2009)]{Nesterov09}
Yurii Nesterov.
\newblock Unconstrained convex minimization in relative scale.
\newblock \emph{Mathematics of Operations Research}, 34\penalty0 (1):\penalty0
  180--193, February 2009.
\newblock ISSN 0364-765X.
\newblock \doi{10.1287/moor.1080.0348}.
\newblock URL \url{http://dx.doi.org/10.1287/moor.1080.0348}.

\bibitem[Ouyang and Gray(2012)]{OuyangG12}
Hua Ouyang and Alexander Gray.
\newblock Stochastic smoothing for nonsmooth minimizations: Accelerating {SGD}
  by exploiting structure.
\newblock In \emph{Proceedings of the 29th International Coference on
  International Conference on Machine Learning}, ICML'12, pages 1523--1530,
  USA, 2012. Omnipress.
\newblock ISBN 978-1-4503-1285-1.
\newblock URL \url{http://dl.acm.org/citation.cfm?id=3042573.3042768}.

\bibitem[Portnoy(1997)]{Portnoy97}
Stephen Portnoy.
\newblock On computation of regression quantiles: Making the laplacian tortoise
  faster.
\newblock \emph{Lecture Notes-Monograph Series}, 31:\penalty0 187--200, 1997.
\newblock ISSN 07492170.
\newblock URL \url{http://www.jstor.org/stable/4355977}.

\bibitem[Portnoy and Koenker(1997)]{PortnoyK97}
Stephen Portnoy and Roger Koenker.
\newblock The gaussian hare and the laplacian tortoise: computability of
  squared-error versus absolute-error estimators.
\newblock \emph{Statistical Science}, 12\penalty0 (4):\penalty0 279--300, 11
  1997.
\newblock \doi{10.1214/ss/1030037960}.
\newblock URL \url{http://dx.doi.org/10.1214/ss/1030037960}.

\bibitem[Ruszczynski and Syski(1986)]{RS86ssgd}
Andrzej Ruszczynski and Wojciech Syski.
\newblock On convergence of the stochastic subgradient method with on-line
  stepsize rules.
\newblock \emph{Journal of Mathematical Analysis and Applications},
  114\penalty0 (2):\penalty0 512 -- 527, 1986.
\newblock ISSN 0022-247X.
\newblock \doi{http://dx.doi.org/10.1016/0022-247X(86)90104-6}.
\newblock URL
  \url{http://www.sciencedirect.com/science/article/pii/0022247X86901046}.

\bibitem[Sachdeva and Vishnoi(2014)]{Sachdeva2014faster}
Sushant Sachdeva and Nisheeth~K. Vishnoi.
\newblock Faster algorithms via approximation theory.
\newblock \emph{Foundations and Trends{\textregistered} in Theoretical Computer
  Science}, 9\penalty0 (2):\penalty0 125--210, March 2014.
\newblock ISSN 1551-305X.
\newblock \doi{10.1561/0400000065}.
\newblock URL \url{http://dx.doi.org/10.1561/0400000065}.

\bibitem[Sohler and Woodruff(2011)]{SohlerW11}
Christian Sohler and David~P. Woodruff.
\newblock Subspace embeddings for the $\ell_1$-norm with applications.
\newblock In \emph{Proceedings of the Forty-third Annual ACM Symposium on
  Theory of Computing}, STOC '11, pages 755--764, New York, NY, USA, 2011. ACM.
\newblock ISBN 978-1-4503-0691-1.
\newblock \doi{10.1145/1993636.1993736}.
\newblock URL \url{http://doi.acm.org/10.1145/1993636.1993736}.

\bibitem[Tropp(2012)]{Tropp12}
Joel~A. Tropp.
\newblock User-friendly tail bounds for sums of random matrices.
\newblock \emph{Foundations of Computational Mathematics}, 12\penalty0
  (4):\penalty0 389--434, August 2012.
\newblock ISSN 1615-3375.
\newblock \doi{10.1007/s10208-011-9099-z}.
\newblock URL \url{http://dx.doi.org/10.1007/s10208-011-9099-z}.
\newblock Available at http://arxiv.org/abs/1004.4389.

\bibitem[Williams(2012)]{Williams12}
Virginia~Vassilevska Williams.
\newblock Multiplying matrices faster than coppersmith-winograd.
\newblock In \emph{Proceedings of the Forty-fourth Annual ACM Symposium on
  Theory of Computing}, STOC '12, pages 887--898, New York, NY, USA, 2012. ACM.
\newblock ISBN 978-1-4503-1245-5.
\newblock \doi{10.1145/2213977.2214056}.
\newblock URL \url{http://doi.acm.org/10.1145/2213977.2214056}.

\bibitem[Woodruff and Zhang(2013)]{WoodruffZ13}
David~P. Woodruff and Qin Zhang.
\newblock Subspace embeddings and $\ell_p$-regression using exponential random
  variables.
\newblock In \emph{{COLT} 2013 - The 26th Annual Conference on Learning Theory,
  June 12-14, 2013, Princeton University, NJ, {USA}}, pages 546--567, 2013.
\newblock URL \url{http://jmlr.org/proceedings/papers/v30/Woodruff13.html}.

\bibitem[Yang et~al.(2016)Yang, Chow, R{\'e}, and Mahoney]{YangCRM16}
Jiyan Yang, Yin-Lam Chow, Christopher R{\'e}, and Michael~W. Mahoney.
\newblock Weighted {SGD} for $\ell_p$ regression with randomized
  preconditioning.
\newblock In \emph{Proceedings of the Twenty-seventh Annual ACM-SIAM Symposium
  on Discrete Algorithms}, SODA '16, pages 558--569, Philadelphia, PA, USA,
  2016. Society for Industrial and Applied Mathematics.
\newblock ISBN 978-1-611974-33-1.
\newblock URL \url{http://dl.acm.org/citation.cfm?id=2884435.2884476}.

\end{thebibliography}

\appendix
\section{Proof of Theorem~\ref{thm:applyKatyusha}}\label{sec:katyushaproof}

In this section, we prove our primary result, stated in \Cref{thm:applyKatyusha}.
Specifically, we further examine whether our strong initialization distance bound will allow us to improve the running time with black-box accelerated stochastic gradient descent methods.
The first step towards this is to apply smoothing reductions to our objective function.

\subsection{Smoothing the Objective Function and Adding Strong Convexity}

As before, we let $\xx^* = \arg\min_{\xx} \norme{\AA\xx - \bb}$. For clarity, we will borrow some of the notation from \cite{AllenZhuH16} to more clearly convey their black-box reductions.

\begin{definition}
	A function $f(x)$ is $(L,\sigma)$-\textit{smooth-sc} if it is both $L$-\textit{smooth} and $\sigma$-\textit{strongly-convex}.
\end{definition}

\begin{definition}
	An algorithm $\calA(f(x),x_0)$ is a $\textsc{Time}_\calA(L,\sigma)$-\textit{minimizer} if $f(x)$ is  $(L,\sigma)$-\textit{smooth-sc} and $\textsc{Time}_\calA(L,\sigma)$ is the time it takes $\calA$ to produce $x'$ such that $f(x') - f(x^*) \leq \frac{f(x_0) - f(x^*)}{4}$ for any starting $x_0$.
\end{definition}

Allen-Zhu and Hazan assume access to efficient $\textsc{Time}_\calA(L,\sigma)$-\textit{minimizer} algorithms, and show how a certain class of objective functions can be slightly altered to meet the smoothness and strong convexity conditions to apply these algorithms without losing too much in terms of error and runtime.

\begin{theorem}[Theorem C.2 from \cite{AllenZhuH16}]\label{thm:reductionAlgo}
	Consider the problem of minimizing an objective function
	\[f(x) = \frac{1}{n}\sum_{i=1}^n f_i(x) \]
	such that each $f_i(\cdot)$ is $G$-Lipschitz continuous. Let $\xx_0$ be a starting vector such that $f(\xx_0) - f(\xx^*) \leq \Delta$ and $\|\xx_0 - \xx^*\|_2^2 \leq \Theta$. 
	Then there is a routine that takes as input a $\textsc{Time}_\calA(L,\sigma)$-\textit{minimizer}, $\calA$, alongside $f(x)$ and $x_0$, with $\beta_0 = \Delta/G^2$, $\sigma_0 = \Delta/\Theta$ and $T = \log_2(\Delta/\epsilon)$, and produces $x_T$ satisfying $f(x_T) - f(x^*) \leq O(\epsilon)$ in total running time \[\sum_{t=0}^{T-1}\textsc{Time}_{\calA}(2^t/\beta_0,\sigma_0\cdot 2^{-t}).\] 
\end{theorem}

It is then straightforward to show that our objective function fits the necessary conditions to utilize \Cref{thm:reductionAlgo}.

\begin{lemma}\label{lem:minEqualSumLipsFunctions}
	If $\AA$ is IRB, then the function $\norm{\AA\xx - \bb}_1$ can be written as $\frac{1}{n}\sum_{i=1}^n f_i(\xx)$ such that each $f_i(\cdot)$ is $O(\sqrt{nd})$-Lipschitz continuous.
\end{lemma}

\begin{proof}
	By the definition of 1-norm,
	\[\norm{\AA\xx - \bb}_1 = \sum_{i=1}^n\abs{\AA_{i,:}\xx-\bb_i}.\]
	We then set $f_i(\xx) = n \cdot \abs{\AA_{i,:}\xx-\bb_i}$ and the result follows from Lemma~\ref{lem:easyLipscitz}.
\end{proof}

We can then incorporate our objective into the routine from \Cref{thm:reductionAlgo}, along with our initialization of $\xx_0$.

\begin{lemma}\label{lem:applyReduction}
	Let $\calA$ be a $\textsc{Time}_\calA(L,\sigma)$-\textit{minimizer}, along with objective $\norm{\AA\xx -\bb}_1$ such that $\AA$ is IRB and $\xx_0 = \AA^T \bb$, then the routine from \Cref{thm:reductionAlgo} produces
	$\xx_T$ satisfying $f(\xx_T) - f(\xx^*) \leq O(\epsilon)$ in total running time \[ \sum_{t=0}^{T-1}\textsc{Time}_\calA\left(O\left(\frac{nd2^t}{\Delta}\right),O\left(\frac{n\Delta}{d \cdot f(\xx^*)^2 2^t}\right)\right).\] 
\end{lemma}

\begin{proof}
	Lemma~\ref{lem:minEqualSumLipsFunctions} implies that we can apply \Cref{thm:reductionAlgo} where $G = O(\sqrt{nd})$, and Lemma~\ref{lem:easyInit} gives $\Theta = O(\frac{d}{n}f(x^*)^2)$.
	We then obtain $\beta_0 = O(\frac{\Delta}{nd})$ and $\sigma_0 = O(\frac{n\Delta}{d f(x^*)^2})$ and substitute these values in the running time of \Cref{thm:reductionAlgo}.
\end{proof}

\subsection{Applying Katyusha Accelerated SGD}\label{subsec:applykatyusha}

Now that we have shown how our initialization can be plugged into the smoothing construction of \cite{AllenZhuH16}, we simply need an efficient $\textsc{Time}_\calA(L,\sigma)$-\textit{minimizer} to obtain all the necessary pieces to prove our primary result.

\begin{theorem}[Corollary 3.8 in \cite{AllenZhu17}]\label{thm:katyushaRuntime}
	There is a routine that is a $\textsc{Time}_\calA(L,\sigma)$-\textit{minimizer} where $\textsc{Time}_\calA(L,\sigma) = d \cdot O(n + \sqrt{nL/\sigma})$.
\end{theorem}

We can then precondition the matrix to give our strong bounds on the initialization distance of $\xx_0$ from the optimal $\xx^*$, which allows us to apply the smoothing reduction and Katyusha accelerated gradient descent more efficiently. 

\applyKatyusha*
\begin{proof}	
	Once again, by preconditioning with Lemma~\ref{lem:lewisAndRotate} and error $O(\epsilon)$ we obtain a matrix $\tilde{\AA}\UU \in \R^{N \times d}$ and a vector $\tilde{\bb} \in \R^n$ in time $O(nnz(\AA)\log{n} + d^{\omega-1}\min\{d\epsilon^{-2}\log{n},n\})$.
	We utilize the routine in \Cref{thm:katyushaRuntime} as the $\textsc{Time}_\calA(L,\sigma)$-\textit{minimizer} for Lemma~\ref{lem:applyReduction}, and plug the time bounds in to achieve an absolute error of $O(\delta)$ in the preconditioned objective function with the following running time:
	\begin{align*} 
	d \cdot O\left( \sum_{t=0}^{T-1}N + \sqrt{N\cdot\left(\frac{Nd 2^{t}}{\Delta}\right)\left(\frac{d\cdot f(x^*)^2 2^{t}}{N\Delta}\right)}\right) 
	&= d \cdot O\left( N\log{\frac{\Delta}{\delta}} +  \frac{d\sqrt{N}\cdot f(x^*)}{\Delta}\sum_{t=0}^{T-1} 2^{t}\right) \\ &= d \cdot O\left( N\log{\frac{\Delta}{\delta}} + \frac{d\sqrt{N}\cdot f(x^*)}{\delta}\right) .
	\end{align*}

	To achieve our desired relative error of $\epsilon$ we need to set $\delta = O(\epsilon f(\xx^*))$. Technically, this means that the input to gradient descent will require at least a constant factor approximation to $f(\xx^*)$. We will show in \Cref{subsec:binarySearch} that we can assume that we have such an approximation at the cost of a factor of $\log{n}$ in the running time. We assume that $f(\xx^*) \geq n^{-c}$ for some fixed constant\footnote{Note that if $f(\xx^*) = 0$, then our initialization $\xx_0 = \AA^T \bb$ will be equal to $\xx^*$.} $c$ in order to upper bound $\log\frac{\Delta}{\delta}$, which gives a runtime of $O\left(d N\log (n\epsilon^{-1}) + \frac{d\sqrt{N}\cdot }{\epsilon}\right)$.
		
	Here, we used the fact that $N = O(d\epsilon^{-2}\log{n})$, but can also assume that computationally, $N \leq n$, as will be addressed in \Cref{subsec:simulatedSplitting}. This gives a runtime of 
	\[
	O\left(\min\{d^{2.5} \epsilon^{-2} \sqrt{\log {n}} , nd\log (n/\epsilon) + \sqrt{n}d^2\epsilon^{-1} \}\right),
	\] which, combined with our preconditioning runtime (where $\Upsilon$ is a lower order term if we assume the $\epsilon$ is at most polynomially small in $n$) and the factor $\log{n}$ overhead from estimating $f(\xx^*)$ which we address in \Cref{subsec:binarySearch}, gives the desired runtime. Furthermore, since the error in our preconditioning was $O(\epsilon)$, by Lemma~\ref{lem:lewisAndRotate} we have achieved a solution with $O(\epsilon)$ relative error in the original problem.
\end{proof}

\section{Preconditioning with Lewis Weights and Rotation}\label{sec:lewis}

In this section, we show how to precondition a given matrix $\AA \in \R^{n \times d}$ into a ``good" matrix, primarily
using techniques by \cite{cohenpeng}, and will ultimately prove Lemma~\ref{lem:lewisAndRotate}.
Recall that our overall goal was to efficiently transform $\AA$ into a matrix $\tilde{\AA}$ such that the $\ell_1$ norm is approximately maintained for all $\xx$, along with $\tilde{\AA}$ being isotropic and having all row norms approximately equal. 

Accordingly, our preconditioning will be done in the following two primary steps:
\begin{enumerate}
	\item We sample $N = O(d \epsilon^{-2} \log d )$ rows from $\AA$ according to Lewis weights \citep{cohenpeng}
	to construct a matrix $\tilde{\AA} \in \R^{N \times d}$. The guarantees of \cite{cohenpeng} ensure that
	for all $x \in \R^d$, $||\tilde{\AA} x||_1 \approx_{1+\epsilon} \norm{\AA x}_1$ with high probability. We then further show that this sampling scheme gives $\tau_i (\tilde{\AA}) = O(d/N)$ for all $1 \leq i \leq N$ with high probability.
	\item We then find an invertible matrix $\UU$ such that $\tilde{\AA}\UU$ still has the two necessary properties from Lewis weight sampling and is also isotropic.
\end{enumerate}

The matrix $\tilde{\AA}\UU$ then has all the prerequisite properties
to run our $\ell_1$ minimization algorithms, and it only becomes necessary to show that running an $\ell_1$-minimization routine on $\tilde{\AA}\UU$
will help us find an approximate solution to the original problem.

In \Cref{subsec:lewis}, we show that Lewis weight sampling gives a matrix with approximately equal leverage scores.
In \Cref{subsec:rotate}, we find the invertible matrix $\UU$ that makes $\tilde{\AA}\UU$ isotropic while preserving other properties.
In \Cref{subsec:translateToOriginal}, we show that an approximate solution with respect to the preconditioned matrix will give an approximate solution with respect to the original matrix.
Finally, we prove our primary preconditioning result, Lemma~\ref{lem:lewisAndRotate}, in \Cref{subsec:proof}.

Before we do this, the following facts are useful. 

\begin{fact}[Foster's theorem \citep{Foster53}]
	\label{fact:fosters}
	For a matrix $\AA \in \R^{n \times d}$,
	\[\sum_{i=1}^n \tau_i(\AA) = d.\]
\end{fact}

\begin{fact}[Lemma 2 in \cite{CohenLMMPS15}]
	\label{fact:boundLev}
	Given a matrix $\AA$, for all rows $i$,
	\[\tau_i(\AA) = \min_{\AA^T\xx = \AA_{i,:}^T} \norme{\xx}^2.\]
\end{fact}

\subsection{Lewis Weight Sampling gives Approximately Equal Leverage Scores}\label{subsec:lewis}


In this section, we prove that sampling according to Lewis weights gives a matrix with approximately equal leverage scores. 
This proof will largely rely on showing that, up to row rescaling, the sampled matrix $\tilde{\AA}$ is such that $\tilde{\AA}^T\tilde{\AA}$ is spectrally close to $\AA^T\AA$.
This proof will boil down to a standard application of matrix concentration bounds for sampling according to leverage scores.
Our primary lemma in this section will then mostly follow from the following lemma, which will be proven at the end of this section.

\begin{lemma}
	\label{lem:l2Sampling}
	Given a matrix $\AA$ that is sampled according to \Cref{thm:lewisWeights} with error $\epsilon$ and gives matrix $\tilde{\AA}$, then
	\[ \tilde{\AA}^T\tilde{\AA} \approx_{O(1)} \frac{1}{h(n,\epsilon)}\AA^T\LW^{-1}\AA \]
	with high probability.
	
\end{lemma}

Using this, we can prove our key lemma.

\begin{lemma}
	\label{lem:lewisGivesEqualLevScores}
	Given a matrix $\AA \in \R^{n \times d}$ that is sampled according to \Cref{thm:lewisWeights} and gives matrix $\tilde{\AA}$, then for all rows $i$ of $\tilde{\AA}$,
	\[ \tau_i(\tilde{\AA}) \approx_{O(1)} \frac{d}{N}\]
	with high probability.
\end{lemma}

\begin{proof}
	Lemma~\ref{lem:l2Sampling} implies that 
	\[ \tau_i(\tilde{\AA}) = \tilde{\AA}_{i,:}\left(\tilde{\AA}^T\tilde{\AA}\right)^{-1}\tilde{\AA}_{i,:}^T \approx_{O(1)} h(n,\epsilon)\cdot \tilde{\AA}_{i,:}\left(\AA^T\LW^{-1}\AA\right)^{-1}\tilde{\AA}_{i,:}^T 
	\] 
	with high probability. 
	\Cref{thm:lewisWeights} implies that every row $i$ of $\tilde{\AA}$ is simply some row $j$ of $\AA$, scaled by $\frac{1}{\pp_j}$. 
	Therefore, for any row $i$ of $\tilde{\AA}$ we must have 
	\[\tau_i(\tilde{\AA}) \approx_{O(1)} h(n,\epsilon) \cdot \tilde{\AA}_{i,:}\left(\AA^T\LW^{-1}\AA\right)^{-1}\tilde{\AA}_{i,:}^T  =  h(n,\epsilon) \cdot \frac{\AA_{j,:}}{\pp_j}\left(\AA^T\LW^{-1}\AA\right)^{-1}\frac{\AA_{j,:}^T}{\pp_j}.
	\]
	From Definition~\ref{def:lewis} we have 
	\[\lw_j^2 = \AA_{j,:}\left(\AA^T\LW^{-1}\AA\right)^{-1}\AA_{j,:}^T\]
	which along with the fact that $\pp_j \approx_{O(1)} \lw_j \cdot h(n,\epsilon)$ reduces the leverage score to
	\[ \tau_i(\tilde{\AA}) \approx_{O(1)} h(n,\epsilon) \cdot \frac{\lw_j^2}{\pp_j^2} \approx_{O(1)}\frac{1}{h(n,\epsilon)}.
	\]
	Finally \Cref{fact:fosters} gives us that the sum of Lewis weights must be $d$ because they are leverage scores of $\LW^{-1/2}\AA$, which implies $\frac{1}{h(n,\epsilon)} \approx_{O(1)} \frac{d}{N}$ by our definition of $N = \sum_i \pp_i$.
\end{proof}

It now remains to prove Lemma~\ref{lem:l2Sampling}.
The proof follows similarly to the proof of Lemma 4 in \cite{CohenLMMPS15}, except that their leverage score sampling scheme draws each row without replacement, and we need a fixed number of sampled rows with replacement. Accordingly, we will also use the following matrix concentration result from \cite{Harvey12}, which is a variant of Corollary 5.2 in \cite{Tropp12}:

\begin{lemma} [\cite{Harvey12}]
	\label{lem:matrixConcentration}
	Let $\YY_1...\YY_k$ be independent random positive semidefinite matrices of size $d \times d$. Let $\YY= \sum_{i=1}^k \YY_i$, and let $\ZZ = \expec{ }{\YY}$. If $\YY_i \preceq R\cdot\ZZ$ then
	\[ \prob{}{\sum_{i=1}^k \YY_i \preceq \left( 1 - \epsilon \right) \ZZ} \leq de^{\frac{-\epsilon^2}{2R}} \]
	and 
	\[ \prob{}{\sum_{i=1}^k \YY_i \succeq \left( 1 + \epsilon \right) \ZZ} \leq de^{\frac{-\epsilon^2}{3R}}. \]
\end{lemma}

\begin{proof}[Proof of Lemma~\ref{lem:l2Sampling}]
	First, we define $\barA \defeq \LW^{-1/2}\AA$.
	Then, by Definition~\ref{def:lewis}, $\lw_i = \tau_i(\barA)$.
	Since $\LW$ is the diagonal matrix of Lewis weights $\lw$,
	each row of $\barA$ is simply $\barA_{i,:} = \lw_i^{-1/2} \AA_{i,:}$.
	
	By construction of our random $\tilde{\AA}$ in \Cref{thm:lewisWeights} we choose a row $j$ of $\AA$ with probability $\frac{\pp_j}{N}$ and scale by $\frac{1}{\pp_j}$. Therefore, if we let $\YY_i$ be the random variable
	\[
	\YY_i = 
	\begin{cases}
	\frac{\AA_{j,:}\AA_{j,:}^T}{\pp_j^2},& \text{with probability } \frac{\pp_j}{N} \text{ for each } j
	\end{cases}
	\]
then,
	\[
	\YY = \sum_{i=1}^N \YY_i = \sum_{i=1}^N \tilde{\AA}_{i,:}\tilde{\AA}_{i,:}^T = \tilde{\AA}^T\tilde{\AA}.
	\]
	
	Furthermore, we can substitute $\barA_{j,:}\sqrt{\lw_i}$ for $\AA_{j,:}$ and use the fact that $\pp_j \approx_{O(1)} \lw_j \cdot h(n,\epsilon)$ to obtain
	\[
		\frac{\AA_{j,:}\AA_{j,:}^T}{\pp_j^2} \approx_{O(1)} \frac{\barA_{j,:}\barA_{j,:}^T}{\pp_j \cdot h(n,\epsilon)}.
	\] 	
	As a result, we have
	\begin{align*}
	\ZZ = \expec{}{\sum_{i=1}^N \YY_i} &= \sum_{i=1}^N \expec{}{\YY_i} \\
	&\approx_{O(1)} \sum_{i=1}^N \sum_{j=1}^n \frac{\barA_{j,:}\barA_{j,:}^T}{N \cdot h(n,\epsilon)} \\
	&= \frac{1}{h(n,\epsilon)} \sum_{j=1}^n \barA_{j,:}\barA_{j,:}^T = \frac{1}{h(n,\epsilon)}\AA^T\LW^{-1}\AA.
	\end{align*}
	In order to apply Lemma~\ref{lem:matrixConcentration} we need to find $R$ such that $\YY_i \preceq R \cdot \ZZ$, which by our construction of $\YY_i$ requires 
	\[\frac{\AA_{j,:}\AA_{j,:}^T}{\pp_j^2} \preceq R\cdot \ZZ
	\]
	for all $j$. We use our constant factor approximations of $\ZZ$ and $\frac{\AA_{j,:}\AA_{j,:}^T}{\pp_j^2}$ to see that it also suffices to show
	\[
	\frac{\barA_{j,:}\barA_{j,:}^T}{\pp_j \cdot h(n,\epsilon)} \preceq \frac{R}{O(1)} \cdot \frac{1}{h(n,\epsilon)}\barA^T\barA.
	\]
	Given that $\tau_j(\barA) = \lw_j$ and $\pp_j \approx_{O(1)} \lw_j \cdot h(n,\epsilon)$, we have
	\[\frac{\barA_{j,:}\barA_{j,:}^T}{\pp_j \cdot h(n,\epsilon)} \preceq \frac{O(1)\barA_{j,:}\barA_{j,:}^T}{\tau_j(\barA) \cdot h(n,\epsilon)^2}
	\]
	which along with the fact (Equation 12 in the proof of Lemma 4 from \cite{CohenLMMPS15}) that 
	\[ \frac{\barA_{j,:}\barA_{j,:}^T}{\tau_j(\barA)} \preceq \barA^T\barA
	\]
	implies that
	\[\YY_i \preceq \frac{O(1)}{h(n,\epsilon)}\ZZ.
	\]
	By Theorem~\ref{thm:lewisWeights} we know that $h(n,\epsilon) \geq c \epsilon^{-2}\log{n}$ for some constant $c$. Plugging this in for $R$ in Lemma~\ref{lem:matrixConcentration} gives that 
	\[
	\YY \approx_{1 + \epsilon} \ZZ
	\]
	or, substituting our values of $\YY$ and $\ZZ$,
	\[ \tilde{\AA}^T\tilde{\AA} \approx_{O(1)} \frac{1}{h(n,\epsilon)}\AA^T\LW^{-1}\AA \]
	with probability at least $1 - 2de^{-\frac{\epsilon^{-2}}{3R}} \geq 1 - 2de^{-\frac{c\log{n}}{O(1)}} \geq 1 - 2dn^{-c/O(1)}$.
	This implies that the statement in the lemma is true with high probability for $c$ bigger than $O(1)$ (where the $O(1)$ comes from our $\pp_i$ approximation of $\lw_i\cdot h(n,\epsilon)$) and our assumption on $n \geq d$.
\end{proof}

\subsection{Rotating the Matrix to Achieve Isotropic Position}\label{subsec:rotate}

Now that we have sampled by Lewis weights and achieved all leverage scores to be approximately equal, we will show that we can efficiently rotate the matrix into isotropic position while still preserving the fact that all leverage scores are approximately equal.

\begin{lemma}
	\label{lem:rotate}
	If $\UU \in \R^{d \times d}$ is an invertible matrix and $\UU^T\UU = \left(\AA^T\AA\right)^{-1}$ then
	\begin{enumerate}
		\item $\left(\AA\UU\right)^T\left(\AA\UU\right) = \II$.
		\item For all rows $i$,
		$ \tau_i(\AA) = \tau_i(\AA\UU)$.
	\end{enumerate}
	
\end{lemma}

\begin{proof}
	For the first condition, we see that 
	\[\UU^T\AA^T\AA\UU = \II \iff \AA^T\AA = (\UU^T)^{-1}\UU^{-1} \iff (\AA^T\AA)^{-1} = \UU^T\UU.\]
	
	For the second condition, the $i$th row of $\AA\UU$ will be $\AA_{i,:}\UU$,
	which by the definition of leverage scores then gives,
	\begin{align*}
	\tau_i(\AA\UU) & = \AA_{i,:}\UU\left((\AA\UU)^T(\AA\UU)\right)^{-1}\left(\AA_{i,:}\UU\right)^T \\
	& = \AA_{i,:}\UU \UU^{-1}\left(\AA^T\AA\right)^{-1}(\UU^T)^{-1}\UU^T\AA_{i,:}^T \\
	& = \AA_{i,:}\left(\AA^T \AA\right)^{-1}\AA_{i,:}^T \\
	& = \tau_i(\AA). 
	\end{align*}
\end{proof}

It is clear then that we want to rotate our matrix by $\UU$ as above, so it only remains to efficiently compute such a $\UU$.

\begin{lemma}
	\label{lem:rotateRoutine}
	Given a full rank matrix $\tilde{\AA} \in \R^{N \times d}$,
	there is a routine $\textsc{Rotate}$ that can find an invertible $\UU$
	such that $\UU^T\UU = \left(\tilde{\AA}^T\tilde{\AA}\right)^{-1}$
	in time $O(Nd^{\omega-1} + d^{\omega})$.
\end{lemma}

\begin{proof}
	Computing $\tilde{\AA}^T \tilde{\AA}$ can be done in $O(Nd^{\omega-1})$ time using fast matrix multiplication. Inverting $\tilde{\AA}^T\tilde{\AA}$, a $d \times d$ matrix that must have an inverse because $\tilde{\AA}$ is full rank, will require $O(d^{\omega})$ time.
	Finally, we perform a QR-decomposition of $\left(\tilde{\AA}^T\tilde{\AA}\right)^{-1}$ in $O \left(d^{\omega}\right)$ time to obtain our square invertible matrix $\UU$.\footnote{For an invertible matrix $\MM \in \R^{d \times d}$, it is easy to see that $\MM(\MM^T \MM)^{-1/2}$ is an orthonormal basis for $\MM$. We can compute $(\MM^T \MM)^{-1}$ using Schur decomposition in $O(d^{\omega})$ time, and by careful analysis of that algorithm, we can also compute $(\MM^T \MM)^{-1/2}$ in the same amount of time.}
\end{proof}

Lastly, we want to ensure that by rotating our matrix, we can still use an approximate solution to the rotated matrix to obtain an approximate solution of the original matrix.

\begin{lemma}\label{lem:approxStillGoodWithRotation}
	Given a matrix-vector pair $\AA \in \R^{n \times d}, \bb \in \R^{n}$,
	another matrix-vector pair $\tilde{\AA} \in \R^{N \times d}, \tilde{\bb} \in \R^{N}$,
	and an invertible matrix $\UU \in \R^{d \times d}$,
	\[\norm{[\AA,\bb]\xx}_1 \approx_{1+\epsilon} \|[\tilde{\AA},\tilde{\bb}]\xx\|_1 \forall \xx \in \R^{d+1} \iff  \norm{[\AA\UU,\bb]\yy}_1 \approx_{1+\epsilon} \|[\tilde{\AA}\UU,\tilde{\bb}]\yy\|_1 \forall \yy \in \R^{d+1}.
	\]
\end{lemma}

\begin{proof}
	This follows immediately from the fact that for any $\xx$ satisfying the LHS,
	there exists a $\yy$ satisfying the RHS, and vice versa.
	Specifically $\yy_{[1,d]} = \UU^{-1}\xx_{[1,d]}$ and $\yy_{d+1} = \xx_{d+1}$, 
	and equivalently $\UU\yy_{[1,d]} = \xx_{[1,d]}$ and $\yy_{d+1} = \xx_{d+1}$.	
\end{proof}

\subsection{Translating between Preconditioned and Original Matrix Solutions}\label{subsec:translateToOriginal}
Our preconditioning combination of Lewis weights and rotating the matrix gives our desired conditions, specifically an IRB matrix, but it remains to be seen that we can take a solution to this preconditioned matrix and translate it back into an approximate solution of the original matrix. In the following lemma we will show that this is in fact true.

\begin{lemma}\label{lem:objectiveApproxAfterLewis}
	Given a matrix-vector pair $\AA \in \R^{n \times d}, \bb \in \R^{n}$,
	another matrix-vector pair $\tilde{\AA} \in \R^{N \times d}, \tilde{\bb} \in \R^{N}$,
	and an invertible matrix $\UU \in \R^{d \times d}$;
	if
	\[ \big\|[\tilde{\AA}\;\tilde{\bb}]\yy\big\|_1 \approx_{1+\epsilon}\big\|[\AA\;\bb]\yy\big\|_1\]
	for all $\yy \in \R^{d+1}$, and
	if $\tilde{\xx}_{\UU}^*$ minimizes $\|\tilde{\AA}\UU\xx - \tilde{\bb}\|_1$, then for any $\tilde{\xx}$ such that 
	\[ \|\tilde{\AA}\UU\tilde{\xx} - \tilde{\bb}\|_1\leq (1 + \delta) \|\tilde{\AA}\UU\tilde{\xx}_{\UU}^* - \tilde{\bb}\|_1 \]
	we must have
	\[
	\|\AA(\UU\tilde{\xx}) - \bb\|_1 \leq (1 + \epsilon)^2(1 + \delta)\|\AA \xx^* - \bb \|_1 \]
	with high probability.	
\end{lemma}

\begin{proof}
	By assumption we have 
	\[ \big\|[\tilde{\AA}\;\tilde{\bb}]\yy \big\|_1 \approx_{1 + \epsilon}\big\|[\AA\;\bb]\yy\big\|_1\]
	for all $\yy \in \R^{d+1}$, and we can then use Lemma~\ref{lem:approxStillGoodWithRotation} to obtain 
	\[ \big\|[\tilde{\AA}\UU\;\tilde{\bb}]\yy\big\|_1 \approx_{1 + \epsilon}\big\|[\AA\UU\;\bb]\yy \big\|_1\]
	for all $\yy \in \R^{d+1}$. 
	By fixing $\yy$ to be $\begin{pmatrix}
	\xx \\ -1
	\end{pmatrix}$, we get
	\begin{align}\label{eqn:approx1norm}
	& \|\tilde{\AA}\xx - \tilde{\bb}\|_1 \approx_{1+\epsilon}\|\AA\xx - \bb\|_1 &  \forall~ \xx \in \R^{d}, \\
	\label{eqn:approx1normrotated} & \|\tilde{\AA}\UU\xx - \tilde{\bb}\|_1 \approx_{1+\epsilon}\|\AA\UU\xx - \bb\|_1 &  \forall~ \xx \in \R^{d}.
	\end{align}
	
	\Cref{eqn:approx1norm} gives
	\[ \norm{{\AA}\UU\tilde{\xx} - {\bb}}_1\leq (1 + \epsilon)\|\tilde{\AA}\UU\tilde{\xx} - \tilde{\bb}\|_1.  \]
	Using our initial assumption and defining $\tilde{\xx}^* \defeq \UU\tilde{\xx}_{\UU}^*$ then gives us
	\[ \norm{{\AA}\UU\tilde{\xx} - {\bb}}_1\leq (1 + \epsilon)(1 + \delta)\|\tilde{\AA}\tilde{\xx}^* - \tilde{\bb}\|_1.  \]
	Notice that if $\tilde{\xx}_{\UU}^*$ minimizes $\|\tilde{\AA}\UU\xx - \tilde{\bb}\|_1$,
	then $\tilde{\xx}^*$ must minimize $\|\tilde{\AA}\xx - \tilde{\bb}\|_1$ because $\UU$ is invertible.
	Therefore, $\|\tilde{\AA}\tilde{\xx}^* - \tilde{\bb}\|_1 \leq \|\tilde{\AA}\xx^* - \tilde{\bb}\|_1$ and we have 
	\[ \|{\AA}\UU\tilde{\xx} - {\bb}\|_1\leq (1 + \epsilon)(1 + \delta)\|\tilde{\AA}{\xx}^* - \tilde{\bb}\|_1.  \]
	Finally, applying \Cref{eqn:approx1normrotated} gives
	\[ \|{\AA}\UU\tilde{\xx} - {\bb}\|_1\leq (1 + \epsilon)^2(1 + \delta)\|{\AA}{\xx}^* - {\bb}\|_1.
	\]
\end{proof}

\subsection{Proof of Lemma~\ref{lem:lewisAndRotate}}\label{subsec:proof}

We now have all the necessary pieces to prove our primary preconditioning lemma, which we will now restate and prove.

\lewisAndRotate*
\begin{proof}
	From Theorem~\ref{thm:lewisWeights} we have that
	\[ \big\|[\tilde{\AA}\;\tilde{\bb}]\yy\big\|_1 \approx_{1+\epsilon}\big\|[\AA\;\bb]\yy\big\|_1\]
	for all $\yy \in \R^{d+1}$ with high probability. Lemma~\ref{lem:objectiveApproxAfterLewis} then gives
	\[
	\norm{\AA(\UU\tilde{\xx}) - \bb}_1 \leq (1 + \epsilon)^2(1 + \delta)\norm{\AA \xx^* - \bb }_1 \]
	by our assumption on $\tilde{x}$.
	
	Lemma~\ref{lem:l2Sampling} and the assumption that $\AA$ is full rank imply that $\tilde{\AA}$ is full rank with high probability.
	Our use of $\textsc{Rotate}$ to generate $\UU$, such that $\UU^T\UU = (\tilde{\AA}^T\tilde{\AA})^{-1}$, along with Lemma~\ref{lem:rotate}, gives $(\tilde{\AA}\UU)^T\tilde{\AA}\UU = \II$ and also that $\tau_i(\tilde{\AA}\UU) = \tau_i(\tilde{\AA})$ for all $i$.
	\Cref{fact:boundLev} gives $\tau_i(\tilde{\AA}) \leq \tau_i([\tilde{\AA}\;\tilde{\bb}])$, which along with Lemma~\ref{lem:lewisGivesEqualLevScores}, implies $\tau_i(\tilde{\AA}\UU) \leq O(d/N)$ for all $i$.
	Finally, by Definition~\ref{def:levScore} and the fact that $(\tilde{\AA}\UU)^T\tilde{\AA}\UU = \II$, we then have 
	\[\tau_i(\tilde{\AA}\UU) = \bigg\|\left(\tilde{\AA}\UU\right)_{i,:}\bigg\|_2.
	\]
	
	The sampling of $\AA$ is done according to the technique by \cite{cohenpeng}, which requires $O(nnz(\AA)\log{n} + d^{\omega})$ time to obtain the sampling probabilities. Then the actual sampling requires $O(\min\{d\epsilon^{-2}\log{n}, (d\epsilon^{-2}\log{n})^{1/2 + o(1)} + n\log^2{n}\})$-time according to Corollary~\ref{cor:sampleEfficiently} shown in \Cref{subsec:simulatedSplitting}. Computing the invertible matrix $\UU$ for input $\tilde{\AA}$ takes $O(Nd^{\omega-1} + d^{\omega})$ time from Lemma~\ref{lem:rotateRoutine}, and the number of rows of $\tilde{\AA}$ is $N = O(d\epsilon^{-2}\log{n})$.
	Finally, Lemma~\ref{lem:duplicateRowsForAtransposeA} and Corollary~\ref{cor:duplicateRowsRotAndInit} in \Cref{subsec:simulatedSplitting} show that this computation time can also be bounded with $N \leq n$, which then gives our desired runtime.
\end{proof}

\section{Proofs from Section 4}\label{sec:sec4proofs}

In this section we provide the omitted proofs from \Cref{sec:sparsitypreserve}. 
\subsection{Proof of Lemma~\ref{lem:uniformSampling}}
\label{subsec:unifSample}

In this section we reduce the number of rows in $\AA$ by uniform sampling while still preserving certain guarantees. Note that we will ultimately sample from $[\AA\;\bb]$, but for simplicity in notation, we will just use $\AA$ here.


To prove Lemma~\ref{lem:uniformSampling}, we need the following lemma, which states the key fact that the conditions on $\AA$ ensure approximately uniform Lewis weights.

\begin{lemma}[Almost-uniform leverage scores imply almost-uniform Lewis weights]\label{lem:almostUniform}
	Consider a matrix $\AA \in \R^{n \times d}$ such that $\AA^T\AA \approx_{O(1)} \II$ and $\norme{\AA_{i,:}}^2 \approx_{O(1)} d/n$.
	Let $\lw$ denote the $\ell_1$ Lewis weights for $\AA$.
	Then for each row $i$, we have $\lw_i \approx_{O(1)} d/n$.
\end{lemma}

\begin{proof}[Proof of Lemma~\ref{lem:uniformSampling}]
	Note that by Lemma~\ref{lem:almostUniform} we have 
	\[
	\pp_i = \frac{N}{n} \approx_{O(1)} \frac{d \cdot O(\epsilon^{-2}\log{n})}{n}  \approx_{O(1)} \lw_i \cdot O(\epsilon^{-2}\log{n}).
	\]
	Thus, if we use $\pp_i = N/n$ for each $i$ in \Cref{thm:lewisWeights}, we get the first property while avoiding the cost of computing $\pp_i$'s stated in \Cref{thm:lewisWeights}.
	
The second property follows from Lemma~\ref{lem:l2Sampling} in \Cref{sec:lewis}. Specifically, we have 
\[
\tilde{\AA}^T\tilde{\AA} \approx_{O(1)} \frac{1}{O(\epsilon^{-2}\log{n})}\AA^T\LW^{-1}\AA,
\]
which then implies that
\[
\tilde{\AA}^T\tilde{\AA} \approx_{O(1)} \frac{n}{d \cdot O(\epsilon^{-2}\log{n})}\AA^T\AA \approx_{O(1)} \frac{n}{N} \II.
\]

Let $\tau$ denote the leverage scores for $\AA$.
Now, for the third property, it follows from the definition of leverage scores and the second property that
\[
\tau_i(\tilde{\AA}) = \norme{\left(\tilde{\AA}^T\tilde{\AA}\right)^{-1/2}\tilde{\AA}_{i,:}^T}^2 \approx_{O(1)} \norme{\sqrt{N/n}\tilde{\AA}_{i,:}^T}^2.
\]
Furthermore, Lemma~\ref{lem:lewisGivesEqualLevScores} in \Cref{sec:lewis} shows that $\tau_i(\tilde{\AA}) \approx_{O(1)} d/N$.
Factoring this into the equation gives us
\[ 
\norme{\tilde{\AA}_{i,:}^T}^2 \approx_{O(1)} dn/N^2. 
\]
\end{proof}

Now, to prove Lemma~\ref{lem:almostUniform}, we need the following definition and lemma.

\begin{definition}[Definition 5.1 of $\alpha$-almost Lewis weights for $\ell_1$ from \cite{cohenpeng}]
	For a matrix $\AA$, an assignment of weights $\ww$ is $\alpha$-\textit{almost Lewis} if
	\begin{align*}
	\AA_{i,:} (\AA^T \WW^{-1}\AA)^{-1} \AA_{:,i}^T \approx_\alpha \ww_i^2,
	\end{align*}
	where $\WW$ is the diagonal matrix form of $\ww$.
\end{definition}

\begin{lemma}[Definition 5.2 and Lemma 5.3 from \cite{cohenpeng}]\label{lem:LewisStability}
	Any set of $\alpha$-almost Lewis weights satisfy
	\begin{align*}
	\lw_i \approx_{\alpha} \ww_i.
	\end{align*}
\end{lemma}

\begin{proof}[Proof of Lemma~\ref{lem:almostUniform}]
We know that $\AA^T \AA \approx_{O(1)} \II$ and for each row $i$, $\norme{\AA_{i,:}}^2 \approx_{O(1)} d/n$. Then,
\begin{align*}
& \tau_i(\AA) = \AA_{i,:} (\AA^T \AA)^{-1} \AA_{i,:}^T \approx_{O(1)} \AA_{i,:} \AA_{i,:}^T\\
\implies & \tau_i(\AA) \approx_{O(1)} d/n.
\end{align*}

That is, all of the leverage scores are approximately equal. Then we can show that $\ww = (d/n)\ones$, where $\ones$ is the all ones vector. Then,
\begin{align*}
\AA_{i,:} (\AA^T \WW^{-1}\AA)^{-1} \AA_{:,i}^T &= (d/n)\AA_{i,:} (\AA^T  \AA)^{-1} \AA_{:,i}^T \approx_{O(1)} d^2/n^2 = \ww_i^2.
\end{align*}
Thus, $\ww$ is $O(1)$-almost Lewis. The result then follows by Lemma~\ref{lem:LewisStability}.
\end{proof}

\subsection{Proof of Lemma~\ref{lem:exactInit} and Lemma~\ref{lem:cg}}
\label{subsec:initialization_proofs}

To prove Lemma~\ref{lem:exactInit}, we use the following lemma:
\begin{lemma}\label{lem:subordinatenorm}
	Let $\vv \in \R^n$ be a vector with $\norm{\vv}_1 = 1$.
	Then, for a matrix $\AA \in \R^{n \times d}$,
	\begin{align*}
	\norme{\AA^T \vv} \leq \max_i \norme{\AA_{i,:}}.
	\end{align*}
\end{lemma}

\begin{proof}
	\begin{align*}
	\norme{\AA^T \vv} & = \norme{\sum_i \AA_{i,:} \vv_i} \leq \max_i \norme{\AA_{i,:}}.
	\end{align*}
	where the inequality follows by the convexity of $\norm{\cdot}_2$ and since $\sum_i \abs{\vv_i} = 1$,
\end{proof}

\begin{proof}[Proof of Lemma~\ref{lem:exactInit}]
	By our assumptions on $\AA$, we have $\AA^T\AA + \BB= \II$ for some symmetric $\BB$ where $\norme{\BB}\le O(1)$. Since $\xx_0 = \arg\min_{\xx} \norm{\AA\xx - \bb}_2$, we have $\xx_0 = \AA^\dagger \bb = (\AA^T \AA)^{-1} \AA^T \bb$.
	\begin{align*}
	\norme{\xx_0-\xx^*} &= \norme{(\AA^T \AA)^{-1} \AA^T \bb - \xx^*}\\
	&= \norme{(\AA^T \AA)^{-1} \AA^T \bb - (\AA^T \AA)^{-1} \AA^T \AA \xx^*}\\
	&= \norme{(\AA^T \AA)^{-1} \AA^T ( \bb - \AA\xx^*)}.
	\end{align*}
	Let $\vv = (\AA \xx^* - \bb)/\norm{\AA \xx^* - \bb}_1$.
	\begin{align*}
	\norme{\xx_0 -\xx^*} &= \norme{(\AA^T \AA)^{-1} \AA^T \vv}\norm{\bb - \AA\xx^*}_1\\
	&= \norme{(\AA^T \AA + \BB)(\AA^T \AA)^{-1} \AA^T \vv}\norm{\bb - \AA\xx^*}_1\\
	&= \norme{(\II + \BB (\AA^T \AA)^{-1}) \AA^T \vv}\norm{\bb - \AA\xx^*}_1\\
	&\le \norme{\II + \BB (\AA^T \AA)^{-1}}\norme{\AA^T \vv}\norm{\bb - \AA\xx^*}_1.
	\end{align*}
	Now note:
	\begin{align*}
	\norme{\II + \BB (\AA^T \AA)^{-1}}&\le \norme{\II} + \norme{\BB}\norme{(\AA^T \AA)^{-1}}\\
	&\le O(1).
	\end{align*}
	Also, by Lemma~\ref{lem:subordinatenorm} and the assumptions on $\AA$,
	\begin{align*}
	\norme{\AA^T \vv} \leq O(\sqrt{d/n}).
	\end{align*}
	Thus, we have:
	\begin{align*}
	& \norme{\xx_0 - \xx^*} \le O(\sqrt{d/n}) \norm{\bb - \AA\xx^*}_1. 
	\end{align*}
\end{proof}

To prove Lemma~\ref{lem:cg}, we use the following theorem from \cite{Sachdeva2014faster}:

\begin{theorem}[Theorem 9.1 from \cite{Sachdeva2014faster}]\label{thm:cg}
	Given an symmetric positive definite matrix $\MM \in \R^{n \times n}$ and a vector $\yy\in \R^{n}$, the Conjugate Gradient method can find a vector $\xx$ such that $\norm{\xx - \MM^{-1} \yy}_{\MM} \le \delta \norm{\MM^{-1} \yy}_{\MM}$ in time $O( (t_{\MM} + n) \cdot \sqrt{\kappa(\MM)} \log (1/\delta))$, where $t_{\MM}$ is the time required to multiply $\MM$ with a given vector and $\kappa(\MM)$ is the condition number of $\MM$.
\end{theorem}

\begin{proof}[Proof of Lemma~\ref{lem:cg}]
Let $\MM = \AA^T \AA$ and $\yy = \AA^T \bb$. Then by \Cref{thm:cg},
the conjugate gradient method finds a vector $\tilde{\xx}_0$ such that
$\norm{\tilde{\xx}_0}_{\AA^T \AA} \le \delta \norm{(\AA^T \AA)^{-1} \AA^T \bb}_{\AA^T \AA}$
in time $O((t_{\AA^T \AA} + d) \log (1/\delta))$.
Noting that $\AA^T \AA \approx_{O(1)} \II$, we get
\begin{align*}
\norme{\tilde{\xx}_0 - \xx_0} \le  O(\delta) \norme{\xx_0}.
\end{align*}

Next, we note:
\begin{align*}
& \norme{\bb} \ge \norme{\AA \xx_0 - \bb} \ge \norme{\AA \xx_0} - \norme{\bb} \\
\implies & \norme{\xx_0} \le O(\norme{\bb}).
\end{align*}

Now, since we assume that $\norme{\bb} \le n^c$ and $\norme{\AA \xx_0 - \bb} \ge 1/n^c$ for some $c$,
we can set $\delta = O(\epsilon/(n^c))$ to get:
\begin{align*}
& \norme{\tilde{\xx}_0 - \xx_0} \le \epsilon/n^c \le \epsilon \norme{\AA\xx_0 - \bb}. 
\end{align*}
\end{proof}

\subsection{Proof of \Cref{thm:noFastMatrixMult}}\label{subsec:achievingSparsity}


\noFastMatrixMult*

\begin{proof}
	We first prove correctness.
	We apply Lemma~\ref{lem:uniformSampling} to $[\AA\;\bb]$ and rescale $[\tilde{\AA}\;\tilde{\bb}]$ by $\sqrt{N/n}$.
	Note that rescaling does not change the relative error of our output $\tilde{x}$.
	From this rescaling, we have $\tilde{[\AA\;\bb]}^T\tilde{[\AA\;\bb]} \approx_{O(1)} \II$ and thus $\tilde{\AA}^T\tilde{\AA} \approx_{O(1)} \II$.
	This implies that $\tau_i([\tilde{\AA}\;\tilde{\bb}]) \approx_{O(1)} \|[\tilde{\AA}\;\tilde{\bb}]_{i,:}\|_2^2$ and $\tau_i(\tilde{\AA}) \approx_{O(1)} \|\tilde{\AA}_{i,:}\|_2^2$.
	By \Cref{fact:boundLev}, we have $\tau_i(\tilde{\AA}) \leq \tau_i([\tilde{\AA}\;\tilde{\bb}])$, so we have $\|\tilde{\AA}_{i,:}\|_2^2 \leq O(d/N)$ for all rows $i$.
	As a result, we can find $\tilde{\xx}_0$ according to Lemma~\ref{lem:initMain}.
	The rest of the correctness follows exactly as in the proof of \Cref{thm:applyKatyusha}.
	
	We now examine the running time and note that uniform sampling will take $O(nnz(\AA) + d^2\epsilon^{-2}\log{n})$ time to produce $[\tilde{\AA}\;\tilde{\bb}]$.
	By Lemma~\ref{lem:initMain}, we can then find $\tilde{\xx}_0$ in time $O(d^2\epsilon^{-2}\log{n})$ because $\tilde{\AA}$ is a $d\epsilon^{-2}\log{n} \times d$ matrix, so $t_{\tilde{\AA}^T\tilde{\AA}} = O(d^2\epsilon^{-2}\log{n})$.
	Finally from the analysis of \Cref{thm:applyKatyusha} we know that accelerated stochastic gradient descent requires $O(d^{2.5}\log^{1/2}{n}\cdot \epsilon^{-2})$ time.
	However, we note that the extra factor of $d$ came from \Cref{thm:katyushaRuntime} where we substituted $d$ for the time per iteration of stochastic gradient descent.
	This value can actually be upper bounded by the maximum number of entries in any row of $\tilde{\AA}$, which because of our uniform sampling is upper bounded by the maximum number of entries in any row of $\AA$.
	Adding a runtime overhead of $\log{n}$ for computing an approximation of the optimal objective, as in \Cref{subsec:binarySearch}, gives the desired runtime.
\end{proof}

\section{Secondary technical details for our main results}\label{sec:minor_details}

In this section, we address a couple assumptions that were made in the proofs of our main results. These assumptions were minor details, but we now include proofs for completeness.
First, we always assumed that our row dimension after preprocessing was $O(\min(n,d\epsilon^{-2}\log n))$, and we will address this in \Cref{subsec:simulatedSplitting}. Second, we required a constant factor approximation of the optimal objective value for which we give the procedure in \Cref{subsec:binarySearch}.
\subsection{Simulated Sampling of $\AA$}\label{subsec:simulatedSplitting}

In Lemma~\ref{lem:lewisAndRotate}, our primary preconditioning lemma, we set $N$ to be the minimum of $n$ and $O(d\epsilon^{-2}\log{n})$.
However, all of our sampling above assumed that $O(d\epsilon^{-2}\log{n})$ rows were sampled to achieve certain matrix concentration results.
Accordingly, we will still assume that $O(d\epsilon^{-2}\log{n})$ rows are sampled, but show that we can reduce the computational cost of any duplicate rows to $O(1)$, and hence the computation factor of $N$ can be assumed to be $\min\{n,O(d\epsilon^{-2}\log{n})\}$. The sampling procedure itself can be done in about $O(n)$ time. At the end of this section, we explain how the running time of Katyusha can be made to depend on $n$, rather than $d\epsilon^{-2}\log n$.

Ultimately, our proof of Lemma~\ref{lem:lewisAndRotate} will critically use the fact that $\tilde{\AA}$ has $O(d\epsilon^{-2}\log{n})$ rows in several places.
The following lemmas will then show how we can reduce this computation for duplicate rows, allowing us to substitute $n$ for $O(d\epsilon^{-2}\log{n})$ in the running time when $n \ll d\epsilon^{-2}\log{n}$.

\begin{lemma}\label{lem:duplicateRowsForAtransposeA}
	Let $\tilde{\AA} \in \R^{N \times d}$ be a matrix with at most $n$ unique rows, and for each unique row, we are given the number of copies in $\tilde{\AA}$. Then computing $\tilde{\AA}^T\tilde{\AA}$ takes at most $O(nd^{\omega -1})$ time.
\end{lemma}

\begin{proof}
	By definition 
	\[
	\tilde{\AA}^T\tilde{\AA} = \sum_i \tilde{\AA}_{i,:}^T\tilde{\AA}_{i,:}.
	\]
	Therefore, if we have $k$ copies of row $\tilde{\AA}_{i,:}$, we know that they contribute $k\tilde{\AA}_{i,:}^T\tilde{\AA}_{i,:}$ to the summation. Accordingly, if we replaced all of them with one row $\sqrt{k}\tilde{\AA}_{i,:}$, then this row would contribute an equivalent amount to the summation. As a result, we can combine all copies of unique rows to achieve an $n \times d$ matrix $\tilde{\AA'}$ and compute $\tilde{\AA'}^T\tilde{\AA'}$ which will be equivalent to $\tilde{\AA}^T\tilde{\AA}$.
\end{proof}

\begin{corollary}\label{cor:duplicateRowsRotAndInit}
	Let $[\tilde{\AA}\; \tilde{\bb}] \in \R^{N \times (d+1)}$ be a matrix with at most $n$ unique rows, and for each unique row, we are given the number of copies in $[\tilde{\AA}\; \tilde{\bb}]$. Then computing $\tilde{\AA}\UU$ where $\UU \in \R^{d\times d}$, and computing $\tilde{\AA}^T\tilde{\bb}$ takes $O(nd^{\omega - 1})$ and $O(nd)$ time, respectively.
\end{corollary}

\begin{proof}
	We can similarly use the fact that $\tilde{\AA}_{i,:}\UU$ is equivalent for all copies of $\tilde{\AA}_{i,:}$ and combine $k$ copies into the row $k\tilde{\AA}_{i,:}$.
	
	Analogously, we have $\tilde{\AA}^T\tilde{\bb} = \sum_i\tilde{\AA}_{i,:}^T\tilde{\bb}_i$, so we can combine duplicate rows.
\end{proof}

Furthermore, we need to show that we can efficiently sample $O(d\epsilon^{-2}\log{n})$ rows (ideally in $O(n)$-time) even when $O(d\epsilon^{-2}\log{n}) \gg n$. We will achieve this through known results on fast binomial distribution sampling.

\begin{theorem}[Theorem 1.1 in \cite{Farach-Colton2015}]\label{thm:binSampling}
	Given a binomial distribution $B(n,p)$ for $n \in \mathbb{N}$, $p \in \mathbb{Q}$, drawing a sample from it takes $O(\log^2{n})$ time using $O(n^{1/2 + \epsilon})$ space w.h.p., after $O(n^{1/2+\epsilon})$-time preprocessing for small $\epsilon > 0$. The preprocessing does not depend on $p$ and can be used for any $p'\in \mathbb{Q}$ and for any $n' \leq n$.
	
\end{theorem}

This result implies that sampling $m$ items independently can be done more efficiently if $m \gg n$, where we are only concerned with the number of times each item in the state space is sampled.
\begin{corollary}\label{cor:binSampling}
	Given a probability distribution $\mathcal{P} = (p_1,...,p_n)$ over a state space of size $n$, sampling $m$ items independently from $\mathcal{P}$ takes $O(m^{1/2 + \epsilon} + n\log^2{n})$-time.
	
\end{corollary}

\begin{proof}
	Note that sampling independently $m$ times is equivalent to determining how many of each item is sampled by using the binomial distribution and updating after each item. More specifically, we can iterate over all $i \in [n]$ and draw $k_i \sim B(m,p_i)$, then update $m$ to be $m - k_i$ and scale up each $p_j$ (where $j > i$) by $(1-p_i)^{-1}$. It is straightforward to make the scaling up of each $p_j$ efficient, and according to Theorem~\ref{thm:binSampling} we can obtain the binomial sample in $O(\log^2{n})$-time.

	Furthermore, because $m$ is decreasing at each iteration, we can use the original preprocessing in Theorem~\ref{thm:binSampling} for each step to achieve our desired running time.	
\end{proof}

\begin{corollary}\label{cor:sampleEfficiently}
	Given a matrix $\AA \in \R^{n \times d}$, with a probability distribution over each row, we can produce a matrix
	$\tilde{\AA} \in \R^{O(d\epsilon^{-2}\log{n})\times d}$ according to the given distribution in time at most
	$O\left(\min(d\epsilon^{-2}\log{n}, (d\epsilon^{-2}\log{n})^{1/2 + o(1)} + n\log^2{n})\right).$
\end{corollary}

Finally, our application of \Cref{thm:katyushaRuntime} assumes that it is given an $N \times d$ matrix, but we assumed that the computational cost could assume $N = \min\{n,O(d\epsilon^{-2}\log{n})\}$. A closer examination of Algorithm 2 in \cite{AllenZhu17}, which is the routine for \Cref{thm:katyushaRuntime}, shows that the factor of $N$ comes from a full gradient calculation, which can be done more quickly by combining rows in an equivalent manner to the lemma and corollary above.

\subsection{Approximating the Optimal Objective Value}\label{subsec:binarySearch}

For ease of notation, we let $f^* = \min_{\xx} \norm{\AA\xx - \bb}_1$ and $f_2^* = \min_{\xx} \norme{\AA\xx - \bb}$ in this section.
In our proof of both Theorem~\ref{thm:ourSGD} and ~\ref{thm:applyKatyusha}, we assumed access to a constant approximation of $f^*$ with a runtime overhead of $\log{n}$.
We will obtain access to this value by giving polynomially approximate upper and lower bounds on $f^*$ and using our primary algorithm on $\log{n}$ guesses for $f^*$ within this range.
We start with the following lemma that gives upper and lower bounds on $f^*$:

\begin{lemma}\label{lem:approxOfObjectiveVal}
	Given a matrix $\AA \in \R^{n \times d}$ and a vector $\bb \in \R^n$, if $\xx_2^*$ minimizes $\norme{\AA\xx-\bb}$ then
	\[ 
	\norme{\AA\xx_2^* - \bb} \leq \norm{\AA\xx^* - \bb}_1 \leq \sqrt{n}\norme{\AA\xx_2^* - \bb}.
	\]
\end{lemma}

\begin{proof}
	By known properties of $\ell_1$ and $\ell_2$, for any $\xx \in \R^n$, we have $\norm{\xx}_2 \leq \norm{\xx}_1 \leq \sqrt{n}\norm{\xx}_2$.
	Accordingly, we must have 
	\[
	\norme{\AA\xx_2^* - \bb} \leq \norme{\AA\xx^* - \bb} \leq \norm{\AA\xx^* - \bb}_1,
	\]
	where the first inequality follows from $\xx_2^*$ being the $\ell_2$-minimizer. Similarly, we also have 
	\[
	\norm{\AA\xx^* - \bb}_1 \leq \norm{\AA\xx_2^* - \bb}_1 \leq \sqrt{n}\norme{\AA\xx_2^* - \bb},
	\]
	where the first inequality follows from $\xx^*$ being the $\ell_1$-minimizer.
\end{proof}

Since $\AA^T \AA = \II$, our initialization of $\AA^T b$ is equal to $\xx_2^*$. Then we can compute $\norme{\AA\xx_2^* - \bb}$ in $O(Nd)$ time.\footnote{Note that our sampled and rotated $\tilde{\AA}\UU$ from Lemma~\ref{lem:lewisAndRotate} loses any sparsity guarantees that $\AA$ may have had.} Consequently, if we let $f_2^*$ be the minimized objective function $\norme{\AA\xx - \bb}$, we can compute polynomially close upper and lower bounds, $f_2^*$ and $\sqrt{n}f_2^*$ respectively, for $f^*$.

\begin{lemma}\label{lem:binarySearch}
	In both variants of our primary algorithm for Theorem~\ref{thm:ourSGD} and ~\ref{thm:applyKatyusha}, we can run the respective algorithms with a constant approximation of $f^*$ by running them $\log{n}$ times using different approximations of $f^*$, which we will denote by $\tilde{f}^*$.
	Furthermore, the runtime of each is independent of the choice of $\tilde{f}^*$.
\end{lemma}

\begin{proof}
	We first examine the latter claim and note that the gradient descent portion of both algorithms take upper bounds on $\norme{\xx_0 - \xx^*}$ as inputs.
	Therefore, given a certain $\tilde{f}^*$ we can input the upper bound $O(\sqrt{d/n})\tilde{f}^*$ and following the analysis of the proofs in Theorems~\ref{thm:ourSGD} and ~\ref{thm:applyKatyusha},
	in runtime $O(d^3\epsilon^{-2})$ and $O\left(d^{2.5}\epsilon^{-2} \sqrt{\log{n}}\right)$ respectively,
	we are guaranteed that we achieve $\tilde{\xx}$ such that $f(\tilde{\xx}) - f^* \leq \epsilon \tilde{f}^*$ with high probability.
	However, note that this is only true if $f^* \leq \tilde{f}^*$.
	Otherwise, we are given no guarantee on the closeness of $f(\tilde{\xx})$ to $f^*$.
	
	The runtime of each algorithm is then not affected by our approximation of $\tilde{f}^*$, however, the closeness guarantees are affected.
	Accordingly, we will run the gradient descent procedure in each respective algorithm $\log {n}$ times with $\tilde{f}^* = f_2^* \cdot 2^i$ for $i = 0$ to $\log {n}$, and whichever iteration produces $\tilde{\xx}$ that minimizes $f(\cdot)$ will be output. Lemma~\ref{lem:approxOfObjectiveVal} implies that there must exist some $i$ such that $f_2^* \cdot 2^i \leq  f^* \leq f_2^* \cdot 2^{i+1}$. Therefore, when we run our algorithm with $\tilde{f}^* = f_2^* \cdot 2^{i+1}$, the algorithm will succeed with high probability. Thus, the overall success probability is at least as high as any individual run of the algorithm. Moreover, the output $\tilde{\xx}$ is guaranteed to have $f(\tilde{\xx}) - f^* \leq 2\epsilon {f}^*$.
\end{proof}

\end{document}